\documentclass[12pt]{article}%
\usepackage[nosort]{cite}
\usepackage{graphicx}
\usepackage{multicol}
\usepackage{amsfonts}
\usepackage{amsmath}
\usepackage{amssymb}
\usepackage{amsthm}
\usepackage{heck}
\usepackage{afterpage}
\usepackage{setspace}
\usepackage{verbatim}
\usepackage{color}
\usepackage{longtable}
\usepackage{float}
\usepackage{subcaption}
\usepackage{epsfig}
\usepackage{epstopdf}
\usepackage{adjustbox}
\usepackage{tikz}
\usepackage[margin=1in]{geometry}
\usepackage{titletoc}
\usepackage{hyperref}%
\usepackage{mathrsfs}
\usepackage{dsfont}
\usepackage{algorithm}
\usepackage{algpseudocode}
\usepackage{tikz}
\usepackage{pgfplots}
\usepackage{algorithmicx, algpseudocode,algorithm}
\usepackage{caption}
\usepackage{float}
\usepackage{makecell}
\usepackage{mathtools}
\usepackage{pdflscape}
\usepackage{array}
\usepackage{longtable}
\usepackage{comment}

\numberwithin{equation}{section}
\newsavebox{\mysavebox}

\hypersetup{colorlinks,citecolor=black,filecolor=black,linkcolor=black,urlcolor=black}

\usepackage[capitalize, noabbrev]{cleveref}
\usetikzlibrary{decorations.markings}
\usetikzlibrary{decorations.markings}
\usetikzlibrary{hobby}
\usepackage{tikz}
\usepackage{tikz-cd}
\usetikzlibrary{arrows}
\usetikzlibrary{arrows.meta}
\usetikzlibrary{arrows,decorations.pathmorphing}
\usetikzlibrary{shapes.geometric,calc,arrows, positioning,shapes.misc,decorations.markings}
\tikzset{
  big arrow/.style={
    decoration={markings,mark=at position 1 with {\arrow[scale=2,#1]{>}}},
    postaction={decorate},
    shorten >=0.4pt},
  big arrow/.default=black}

\pgfdeclarelayer{edgelayer}
\pgfdeclarelayer{nodelayer}
\pgfsetlayers{edgelayer,nodelayer,main}
\pgfplotsset{compat=1.16}
\tikzstyle{none}=[inner sep=0pt]

\newcommand{\GL}{\mathrm{GL}}
\newcommand{\SL}{\mathrm{SL}}
\newcommand{\Mp}{\mathrm{Mp}}
\renewcommand{\O}{\mathrm{O}}
\newcommand{\SO}{\mathrm{SO}}
\newcommand{\SU}{\mathrm{SU}}
\newcommand{\Spin}{\mathrm{Spin}}
\newcommand{\Pin}{\mathrm{Pin}}
\newcommand{\U}{\mathrm{U}}
\newcommand{\USp}{\mathrm{USp}}

\newcommand{\Sq}{\mathrm{Sq}}
\newcommand{\cA}{\mathcal A}

\newcommand{\pt}{\mathrm{pt}}

\newtheorem{thm}[equation]{Theorem}
\newtheorem{prop}[equation]{Proposition}
\newtheorem{lem}[equation]{Lemma}
\theoremstyle{definition}
\newtheorem{defn}[equation]{Definition}

\crefname{thm}{Theorem}{Theorems}
\crefname{prop}{Proposition}{Propositions}
\crefname{defn}{Definition}{Definitions}
\crefname{lem}{Lemma}{Lemmas}

\tikzstyle{NodeCross}=[draw, shape=circle, cross out, inner sep=0pt, minimum size=6pt,line width=0.25mm]
\tikzstyle{Circle}=[draw, shape=circle, black, fill=black, inner sep=0pt, minimum size=6pt]
\tikzstyle{circle}=[draw, shape=circle, black, fill=black, inner sep=0pt, minimum size=16pt]
\tikzstyle{Star}=[draw, shape=star, fill=black, star points=8, inner sep=0pt, minimum size=8pt]
\tikzstyle{CircleRed}=[draw, shape=circle, black, fill=red, inner sep=0pt, minimum size=6pt]
\tikzstyle{StarP}=[draw={rgb,255: red,128; green,0; blue,128}, shape=star, fill={rgb,256: red,128; green,0; blue,128}, star points=8, inner sep=0pt, minimum size=12pt]
\tikzstyle{ShadedCircRed}=[draw=red, shape=circle, fill={rgb, 255: red,255; green,114; blue, 118}, inner sep=0pt, minimum size=80pt, line width=0.5mm, fill opacity=0.2]
\tikzstyle{ShadedCircRed2}=[draw=red, shape=circle, fill={rgb, 255: red,255; green,114; blue, 118}, inner sep=0pt, minimum size=10pt]
\tikzstyle{ShadedCircRed3}=[draw=black, shape=rectangle, fill={rgb, 255: red,255; green,114; blue, 118}, inner sep=0pt, minimum size=113pt, line width=0.25mm]
\tikzstyle{ShadedCirc}=[draw=red, shape=circle, fill=white, inner sep=0pt, minimum size=45pt,  fill opacity=1.0,  line width=0.5mm]
\tikzstyle{CircleBlue}=[draw, shape=circle, fill=blue, inner sep=0pt, minimum size=6pt]
\tikzstyle{BigCirclePurple}=[draw, shape=circle, fill={rgb,255: red,191; green,0; blue,191}, inner sep=0pt, minimum size=12pt]
\tikzstyle{CirclePurple}=[draw, shape=circle, fill={rgb,255: red,191; green,0; blue,191}, inner sep=0pt, minimum size=5pt]
\tikzstyle{EmptyCircle}=[draw, shape=circle, inner sep=0pt, minimum size=4pt]
\tikzstyle{GreenCircle}=[draw, shape=circle,  fill={rgb,255: red,80; green,200; blue,120}, inner sep=0pt, minimum size=8pt]
\tikzstyle{BrownCircle}=[draw, shape=circle,  fill={rgb,255: red,210; green,105; blue,30}, inner sep=0pt, minimum size=8pt]
\tikzstyle{CirclePurpleSmall}=[draw, shape=circle, fill={rgb,255: red,191; green,0; blue,191}, inner sep=0pt, minimum size=4pt]
\tikzstyle{BigCircleGreen}=[draw, shape=circle, fill={rgb,255: red,0; green,191; blue,0}, inner sep=0pt, minimum size=12pt]
\tikzstyle{BigCircleBlue}=[draw, shape=circle, fill={rgb,255: red,0; green,0; blue,191}, inner sep=0pt, minimum size=12pt]
\tikzstyle{BigCircleRed}=[draw, shape=circle, fill={rgb,255: red,191; green,0; blue,0}, inner sep=0pt, minimum size=12pt]
\tikzstyle{BrownCircleSmall}=[draw, shape=circle,  fill={rgb,255: red,210; green,105; blue,30}, inner sep=0pt, minimum size=6pt]
\tikzstyle{SmallCircleBrown}=[draw, shape=circle,  fill={rgb,255: red,210; green,105; blue,30}, inner sep=0pt, minimum size=5pt]
\tikzstyle{SmallCircleRed}=[draw, shape=circle, fill={rgb,255: red,191; green,0; blue,0}, inner sep=0pt, minimum size=6pt]
\tikzstyle{DashedLine}=[-, densely dashed, line width=0.25mm]
\tikzstyle{DottedLine}=[-, dotted, line width=0.25mm]
\tikzstyle{ThickLine}=[-, line width=0.25mm]
\tikzstyle{ArrowLineRight}=[-, -{Stealth[scale=1.25]}, line width=0.25mm, scale=5]
\tikzstyle{ArrowLineRed}=[-, draw={rgb,255: red,191; green,0; blue,0}, -{Stealth[scale=1.75]}, line width=0.1mm, scale=5]
\tikzstyle{RedLine}=[-, draw={rgb,255: red,191; green,0; blue,0}, fill=none, line width=0.5mm]
\tikzstyle{DashedLineThin}=[-, densely dashed, line width=0.125mm, fill=none, draw=black]
\tikzstyle{DottedRed}=[-, dotted, draw={rgb,255: red,191; green,0; blue,0}, fill=none, line width=0.25mm]
\tikzstyle{DashedRed}=[-, densely dashed, draw={rgb,255: red,191; green,0; blue,0}, fill=none, line width=0.25mm]
\tikzstyle{BlueLine}=[-, draw={rgb,255: red,0; green,0; blue,191}, fill=none, line width=0.5mm]
\tikzstyle{ArrowLineBlue}=[-, draw={rgb,255: red,0; green,0; blue,191}, -{Stealth[scale=1.75]}, line width=0.1mm, scale=5]
\tikzstyle{GreenDoubleArrow}=[<->, draw={rgb,155: red,0; green,255; blue,0},  line width= 0.5mm, scale=5]
\tikzstyle{RedDoubleArrow}=[<->, draw={rgb,255: red,255; green,0; blue,0},  line width= 0.5mm, scale=5]
\tikzstyle{BlueDottedLight}=[-, dotted, draw={rgb,255: red,0; green,0; blue,191}, fill=none, line width=0.3mm]
\tikzstyle{BrownLine}=[-, draw={rgb,255: red,210; green,105; blue,30}, fill=none, line width=0.5mm]
\tikzstyle{DottedRed}=[-, dotted, draw={rgb,255: red,191; green,0; blue,0}, fill=none, dotted, line width=0.5mm]
\tikzstyle{DottedPurple}=[-, dotted, draw={rgb,255: red,191; green,0; blue,191}, fill=none, dotted, line width=0.5mm]
\tikzstyle{BlueDottedLight}=[-, dotted, draw={rgb,255: red,0; green,0; blue,191}, fill=none, line width=0.5mm]
\tikzstyle{ArrowLinePurple}=[-, draw={rgb,255: red,191; green,0; blue,191}, -{Stealth[scale=1.75]}, line width=0.5mm, scale=5]
\tikzstyle{DashedLineGreen}=[-, densely dashed, draw={rgb,255: red,74; green,103; blue,65}, line width=0.25mm]
\tikzstyle{LineGreen}=[-, draw={rgb,255: red, 74; green,200; blue,65}, line width=0.5mm]
\tikzstyle{ArrowLineGreen}=[-, draw={rgb,255: red,0; green,191; blue,0}, -{Stealth[scale=1.75]}, line width=0.5mm, scale=5]
\tikzstyle{GreenLine}=[-, draw={rgb,255: red,0; green,191; blue,0}, fill=none, line width=0.5mm]
\tikzstyle{PurpleLine}=[-, draw={rgb,255: red,191; green,0; blue,191}, fill=none, line width=0.5mm]
\tikzstyle{PPurpleLine}=[-, draw={rgb,255: red,191; green,0; blue,191}, fill=none, line width=2.5mm]
\tikzstyle{DPurpleLine}=[-, dotted, draw={rgb,255: red,191; green,0; blue,191}, fill=none, line width=0.5mm]
\tikzstyle{SBrownLine}=[-, draw={rgb,255: red,191; green,0; blue,191}, fill=none, opacity=0.35, line width=2.5mm]
\tikzstyle{DottedBlue}=[-, dotted, draw=blue, fill=none, dotted, line width=0.5mm]
\tikzstyle{DashedPurpleLine}=[-, densely dashed, draw={rgb,255: red,191; green,0; blue,191}, fill=none, line width=0.5mm]
\tikzstyle{SmallCircleBlue}=[draw, shape=circle, fill=blue, inner sep=0pt, minimum size=5pt]
\tikzstyle{SmallCirclePurple}=[draw, shape=circle, fill={rgb,255: red,191; green,0; blue,191}, inner sep=0pt, minimum size=5pt]
\tikzset{snake it/.style={decorate, decoration=snake}}

\tikzset{
dashstar/.style={
 dash pattern=on 5pt off 5pt,
 postaction={
  decorate,
  decoration={
   markings,
   mark=between positions 9pt and 1 step 10pt with {
     \node[color=red] {*};
   }
  }
 }
},
dashstarstar/.style={ 
 dash pattern=on 5pt off 10pt,
 postaction={
   decorate,
   decoration={
     markings,
     mark=between positions 10pt and 1
          step 15pt
           with {
            \node at (-2pt,0pt) {\pgfuseplotmark{star}};
            \node at (2pt,0pt) {\pgfuseplotmark{star}};
           }
   }
 }
}
}

\begin{document}

\date{September 2025}

\preprint{CERN-TH-2025-180}

\title{Exploring Pintopia: \\[4mm] Reflection Branes, Bordisms, and U-Dualities}

\institution{PENN}{\centerline{$^{1}$Department of Physics and Astronomy, University of Pennsylvania, Philadelphia, PA 19104, USA}}
\institution{KENTUCKY}{\centerline{$^{2}$Department of Mathematics, University of Kentucky, 719 Patterson Office Tower, Lexington, KY 40506-0027}}
\institution{CERN}{\centerline{$^{3}$Theoretical Physics Department, CERN, 1211 Geneva 23, Switzerland}}
\institution{PENNmath}{\centerline{$^{4}$Department of Mathematics, University of Pennsylvania, Philadelphia, PA 19104, USA}}

\authors{
Vivek Chakrabhavi\worksat{\PENN}\footnote{e-mail: \texttt{vivekcm@sas.upenn.edu}},
Arun Debray\worksat{\KENTUCKY}\footnote{e-mail: \texttt{a.debray@uky.edu}}, \\[4mm]
Markus Dierigl\worksat{\CERN}\footnote{e-mail: \texttt{markus.dierigl@cern.ch}}, and
Jonathan J. Heckman\worksat{\PENN,\PENNmath}\footnote{e-mail: \texttt{jheckman@sas.upenn.edu}}
}

\abstract{
The U-dualities of maximally supersymmetric non-chiral supergravity
(SUGRA)\ theories lead to strong constraints on the non-perturbative structure
of quantum gravity. In this paper we
determine Spin- and Pin-lifts of these dualities, which extend this action to
fermionic degrees of freedom. Among other things, this allows us to access
non-supersymmetric sectors of these low energy effective field theories in
which bosonic and fermionic degrees of freedom are treated differently. We use
this refinement of the duality groups, in tandem with the Swampland Cobordism Conjecture, to predict new codimension-two branes. These are a natural generalization of the recently discovered R7-branes of type II string theories. The first bordism groups for Spin-twisted duality bundles follow directly from the Abelianization of the duality groups. Viewing the SUGRA\ theory as the low energy limit of a toroidal compactification of M-theory, winding around these codimension-two defects enacts a reflection around one of the torus
directions, which in the effective field theory appears as a charge
conjugation symmetry. We establish some basic properties of such branes,
including determining BPS objects which can end on it, as well as braiding rules and 
bound states realized by multiple reflection branes.
}

\maketitle

\enlargethispage{\baselineskip}

\setcounter{tocdepth}{2}

\tableofcontents

\newpage

\section{Introduction} \label{sec:Intro}

Dualities provide important insight into the non-perturbative structure of
quantum field theory (QFT) and quantum gravity (QG). When combined with
supersymmetry, this leads to powerful constraints on the spectrum and
properties of BPS\ objects.

A celebrated example of this sort are the U-dualities of maximally
supersymmetric non-chiral supergravity (SUGRA) theories in $D$ spacetime dimensions.\footnote{See \cite{Cremmer:1978ds, Cremmer:1979up, Cremmer:1980gs, Julia:1980gr, deWit:1986mz}.} From a top down perspective, such theories arise from M-theory compactified on
a $T^{d}$ with $D+d=11$, unifying many string dualities \cite{Hull:1994ys, Witten:1995ex}. The resulting duality symmetries combine S-dualities
and T-dualities, as inherited from 11D diffeomorphisms and T-dualities of type
II\ string theories on a $T^{d-1}$.

What happens if we relax the assumption that our branes / defects preserve
supersymmetry? In this case more care is needed both in defining what we mean
by U-dualities, as well as in developing new tools to extract the spectrum
of branes / defects.

Our aim in this work will be concentrated in two complementary directions. On
the one hand, we revisit the structure of the U-duality groups in maximally
supersymmetric theories, showing that there are natural Spin- and Pin$^{+}$-lifts which can act non-trivially on fermionic degrees of freedom. On the
other hand, we shall use this structure to extract additional
non-supersymmetric data on the spectrum of objects. The first goal will be
achieved via a study of possible extensions of the bosonic U-duality groups.
The second goal will be achieved by leveraging the
recently proposed Swampland Cobordism Conjecture \cite{McNamara:2019rup} to
extract qualitatively new codimension-two objects.\footnote{See e.g.,
\cite{McNamara:2019rup, Montero:2020icj, Dierigl:2020lai, McNamara:2021cuo,
Blumenhagen:2021nmi, Buratti:2021yia, Debray:2021vob, Andriot:2022mri,
Dierigl:2022reg, Blumenhagen:2022bvh, Velazquez:2022eco, Angius:2022aeq,
Blumenhagen:2022mqw, Angius:2022mgh, Blumenhagen:2023abk, Debray:2023yrs,
Dierigl:2023jdp, Kaidi:2023tqo, Huertas:2023syg, Angius:2023uqk,
Kaidi:2024cbx, Angius:2024pqk, Fukuda:2024pvu, Braeger:2025kra,
Heckman:2025wqd} for recent work on the Swampland Cobordism Conjecture.}

Recently, it was proposed that taking into account the fermionic degrees of
freedom can lead to subtle extensions of such duality groups \cite{Pantev:2016nze, Tachikawa:2018njr, Debray:2021vob, Debray:2023yrs}.\footnote{See also \cite{Lazaroiu:2021vmb} for additional discussion of U-duality bundles.}
The main reason we need to do this is to specify a choice of bundle assignment for fermions which transform under both Spin and duality transformations. In general, these can be globally correlated. The case that has been studied in the greatest detail is that of the IIB\ duality
group $\SL(2,\mathbb{Z})$ and its corresponding Spin- (see \cite{Pantev:2016nze}%
)\ and $\Pin^+$- (see \cite{Tachikawa:2018njr}) lifts. Much as
$\SU(2)$ is the Spin-lift of $\SO(3)$, there is a non-trivial 
$\mathbb{Z}_{2}$ extension of the 
$\SL(2,\mathbb{Z})$ dualities generated by large diffeomorphisms of a $T^{2}$ to the metaplectic cover $\Mp(2,\mathbb{Z})$. 
In \cite{Pantev:2016nze} some Spin-lifts of other U-duality groups were also given, 
but as far as we are aware, a systematic study of all possible Spin-lifts of U-duality groups was not undertaken. 
One of our goals will be to determine the Spin-lift of U-dualities:
\begin{equation}
1\rightarrow
\mathbb{Z}_{2}\rightarrow \widetilde{G_{U}}\rightarrow G_{U}\rightarrow 1 \,.
\end{equation}

Including reflections along one of the directions of the compactification
torus $T^{d}$ leads to a further generalization where one instead extends the
bosonic duality group by allowing reflections, i.e., allowing orientation
reversing transformations of the internal $T^{d}$ as well. In the low energy
effective field theory this reflection is a charge conjugation symmetry. 
This can be written in terms of a short exact sequence of groups:
\begin{equation}
1 \rightarrow G_U\rightarrow G_U \rtimes \mathbb{Z}_2^R \rightarrow \mathbb{Z}_2^R \rightarrow 1 \,,
\end{equation}
where the $R$ superscript in $\mathbb{Z}_2^{R}$ refers to reflections. 
For the subgroup $\SL(d, \mathbb{Z})$ of large diffeomorphisms acting on $T^{d}$, this extensions yields $\GL(d,\mathbb{Z})$. As found in \cite{Tachikawa:2018njr}, the full duality group of
IIB\ string theory is then the Pin$^{+}$ cover of $\GL(2,\mathbb{Z})$. As far
as we are aware, the Pin$^{+}$ cover of more general U-duality groups has not been considered. This is given by a non-trivial central extension:
\begin{equation}
    1 \rightarrow \mathbb{Z}_2 \rightarrow \widetilde{G_U}^+ \rightarrow (G_U \rtimes \mathbb{Z}_2^R) \rightarrow 1 \,.
\end{equation}

Throughout, we shall refer to the Spin- and Pin$^{+}$-lifts of U-dualities as
$\widetilde{G_{U}}$ and $\widetilde{G_{U}}^{+}$. Due to the correlation
between the Spin structure and the duality bundle, the relevant structure
group for spacetimes is then:
\begin{align}
\Spin\text{-}\widetilde{G_{U}}\text{ structure}  &  \text{: } \quad \frac
{\Spin \times\widetilde{G_{U}}}{\mathbb{Z}_{2}}\label{eq:twisted}\\
\Spin\text{-}\widetilde{G_{U}}^{+}\text{ structure}  &  \text{: } \quad%
\frac{\Spin \times\widetilde{G_{U}}^{+}}{
\mathbb{Z}_{2}}, \label{eq:moretwisted}
\end{align}
where the $\mathbb{Z}_2$ in the quotient embeds as $(-1)^F$ in the $\Spin$ factor and the $\mathbb{Z}_2$ of the central extension in the U-duality group.

This extension of the duality groups provides additional access to the
non-perturbative as well as non-supersymmetric sectors of these theories. In
this work we use the Swampland Cobordism Conjecture \cite{McNamara:2019rup} to extract some general predictions for new non-supersymmetric branes in $D$-dimensional supergravity
theories. The main point is that while the Swampland Cobordism Conjecture
asserts that the bordism groups of quantum gravity are trivial:
\begin{equation}
\Omega_{k}^{\text{QG}}=0 \,\text{,}%
\end{equation}
in practice, actual bordism groups are often non-trivial!
As such, the Cobordism Conjecture predicts that the low energy effective field theory must be
supplemented by additional degrees of freedom. Note also that these new
objects are necessarily singular in the low energy effective field theory, and
are also stable against deformations to configurations describable in the supergravity theory at low energies \cite{Kaidi:2024cbx, Heckman:2025wqd}.

For our purposes, the relevant bordism groups are $\Omega_{k}^{\Spin\text{-}%
\widetilde{G_{U}}}(\pt)$ and $\Omega_{k}^{\Spin\text{-}\widetilde{G_{U}}^{+}%
}(\pt)$; namely, those in which the Spin structure and duality bundle are
correlated, as in lines \eqref{eq:twisted} and \eqref{eq:moretwisted}. While the computation of these bordism groups is in general quite difficult, in the case of low values of
$k$ much more can be said without much machinery. In particular, for the codimension-two objects of the theory, we prove that the bordism group in question is captured by
just the Abelianization of the duality bundle structure group:
\begin{equation}
\Omega_{1}^{\Spin\text{-}\widetilde{G_{U}}^{(+)}}(\text{pt})=\text{Ab}%
\big[\widetilde{G_{U}}^{(+)}\big] \,.
\end{equation}
In the case of $D=8,\,\textrm{and}\,\,9$, where $G_{U}$ contains a simple $\SL(2,\mathbb{Z})$ factor, the 
Abelianization of $\widetilde{G_{U}}^{+}$ is $\mathbb{Z}_{2}\oplus \mathbb{Z}_{2}$, where one of the factors is generated by a supersymmetric brane in codimension-two, and the other is a non-supersymmetric brane induced from a reflection on the internal torus.
In the case of $D\leq7$ we find that the only surviving
element of the bordism group is a $\mathbb{Z}_{2}$, which geometrically descends from a reflection of a single direction in
the internal $T^{d}$. Other reflections are obtained by conjugating with the
duality group.

This codimension-two object is simply the $D$-dimensional version of the
reflection 7-branes (R7-branes) found\footnote{See also hints of this brane in \cite{Distler:2009ri}.} 
in \cite{Dierigl:2022reg, Debray:2023yrs} and
further studied in \cite{Dierigl:2023jdp, Ruiz:2024jiz, Heckman:2025wqd}.
Indeed, starting from the reflection 7-brane associated with a $(-1)^{F_{L}}$
monodromy transformation, we observe that compactification on additional circles
leads to the corresponding objects in the lower-dimensional theory. This reflection brane is non-supersymmetric 
since it does not preserve any Killing spinors. Additionally, using the same reasoning presented in \cite{Kaidi:2024cbx, Heckman:2025wqd}, 
it is stable against deformations to any smooth field configurations in the low energy effective field theory.

Much as in earlier studies of the R7-brane, we can also use simple topological
arguments to characterize some properties of reflection branes. For one, we
immediately see that BPS branes can terminate on the reflection brane of the $D$-dimensional theory. 
Additionally, since the R7-brane serves as a
collapsed bubble configuration for a IIA/IIB wall (with a $(-1)^{F_L}$ monodromy cut),\footnote{See \cite{Heckman:2025wqd}.}
we also see that analogous statements hold for these branes as well, separating the same
configurations wrapped on additional directions of an internal $T^{d}$.

Similar considerations hold for the dynamics of multiple reflection branes. We
find that for reflections associated with an even number of distinct internal directions, the
resulting branes form supersymmetric bound states characterized by a local
geometry of the form $T^{d-2k}\times\left(\mathbb{C} \times T^{2k}\right)  / \mathbb{Z}_{2}^{k}$, where the $\mathbb{Z}_{2}^{k}$ act on pairs of holomorphic coordinates. In the case of an odd number of distinct reflections, we instead arrive at a non-supersymmetric configuration, namely a supersymmetric background with an additional supersymmetry breaking reflection brane added in. This geometry can be written as $T^{d-(2k+1)}\times\text{Cone}(\mathrm{KB}_{\infty} \times T^{2k})/\Gamma$, where the $\mathrm{KB}_{\infty}$ is a circle bundle over the spacetime $S^{1}_{\infty} = \partial \mathbb{C}$ resulting in a Klein bottle (with Pin$^{+}$ structure).

The rest of this paper is organized as follows. In Section \ref{sec:GROUP} we
determine the $\Spin$- and $\Pin^{+}$-lifts of the bosonic duality groups. In
Section \ref{sec:BORDISM} we compute the one-dimensional bordism groups associated to those $\Spin$- and $\Pin^+$-lifts, and in Section \ref{sec:REFLECTIONBRANES} we
derive some properties of the reflection branes predicted by the Cobordism
Conjecture. We present our conclusions in Section \ref{sec:CONC}. We discuss some additional features of reflections on the massive and massless spectrum of M-theory in Appendix \ref{app:MASSIVE}. This also provides a complementary perspective on why we must enlarge the U-duality group. Some additional technical details on the standard ``bosonic'' U-duality groups are reviewed in Appendix \ref{app:Split_vs_Compact}. In Appendix \ref{app:Spin_Pin_Lifts} we discuss the Spin- and Pin$^{+}$-lifts of these duality groups to $\widetilde{G_U}$ and $\widetilde{G_{U}}^{+}$. Appendices \ref{app:LHS} and \ref{app:Adams} discuss some further aspects of the corresponding group (co)homology and bordism groups using spectral sequence techniques.

\section{Spin- / \texorpdfstring{Pin$^+$}{Pin+}-Lifts of U-Duality Groups} \label{sec:GROUP}

In this paper we focus on the duality groups of maximally supersymmetric non-chiral supergravity theories. The $D$-dimensional effective supergravity theory can originate both from M-theory compactifications on $\mathbb{R}^{D-1,1} \times T^d$ or type II string theory on $\mathbb{R}^{D-1,1} \times T^{d-1}$. The symmetry / duality group of the theory is known as the U-duality group, and is composed of both strong / weak dualities as well as T-dualities. To be precise, the \emph{bosonic} U-duality group takes the form:\footnote{In \cite{Obers:1998fb} this is stated as $\SL(d,\mathbb{Z}) \bowtie \SO(d-1,d-1,\mathbb{Z})$, but taking into account the RR states and their transformation as spinors of the T-duality group requires this slight refinement.}
\begin{equation}
    G_U = \SL(d,\mathbb{Z}) \bowtie \Spin(d-1,d-1,\mathbb{Z}) \,,
\end{equation}
where the components correspond to the group of large diffeomorphisms on the M-theory torus and the group of T-dualities in type II string theories respectively.\footnote{The classical U-duality group is defined over $\mathbb{R}$, but is broken to the discrete group over $\mathbb{Z}$ due to quantum effects \cite{Hull:1994ys}.} Here, it is important to note that the T-duality group must include spinor representations since the RR fields transform in bispinor representations of:
\begin{equation}
\frac{\Spin(d-1, \mathbb{Z})_{L} \times \Spin(d-1, \mathbb{Z})_{R}}{\mathbb{Z}_2} \subset \Spin(d-1,d-1,\mathbb{Z}),
\end{equation}
namely the left- and right-moving spinors. That being said, observe that purely left-moving or right-moving spinors do not naturally lift to representations of $\Spin(d-1,d-1,\mathbb{Z})$, a fact we will need to handle with care.

In any case, from the perspective of low energy supergravity, the relevant U-duality group is $G_{U}(\mathbb{R})$, namely the split real form of a certain Lie group. The refinement to M-theory imposes a quantization condition to $G_{U} \equiv G_{U}(\mathbb{Z})$, in accord with a discrete spectrum of branes / non-perturbative objects \cite{Hull:1994ys, Witten:1995ex}. For $3 \leq D \leq 9$,
the U-duality groups are:
\begin{equation}
\centering
\begin{tabular}{|c|c|}
\hline
\textbf{$D$} & \makecell{\textbf{Bosonic U-duality}\\\textbf{Group $G_U$}} \\
\hline
$9$ & $\SL(2,\mathbb{Z})$  \\
$8$ & $\SL(3,\mathbb{Z}) \times \SL(2,\mathbb{Z})$ \\
$7$ & $\SL(5,\mathbb{Z})$ \\
$6$ & $\Spin(5,5,\mathbb{Z})$ \\
$5$ & $E_{6(6)}(\mathbb{Z})$ \\
$4$ & $E_{7(7)}(\mathbb{Z})$ \\
$3$ & $E_{8(8)}(\mathbb{Z})$ \\
\hline
\end{tabular}
\label{tab:Summary_Bos_U-dualities}
\end{equation}
See Appendix \ref{app:Split_vs_Compact} for additional discussion of the U-duality groups.

The main subtlety we need to contend with is the behavior of fermions under toroidal compactification, and in particular, the interplay with U-dualities. Locally, one can characterize fermions as transforming in spinor representations, and this clearly descends from 11D. Globally, however, we need to also specify how such objects transform in duality bundles.

Consider, for example, an effective field theory compactified on a $T^d$. Now, in the case of supersymmetric compactifications to a spacetime with Spin structure, the spinors will transform in a section of a bundle with structure group $\Spin(D-1,1) \times \SL(d,\mathbb{Z})$. We can clearly consider more general U-duality groups $\Gamma$, and we characterize this by a connection of the schematic form:
\begin{equation}
\mathcal{D} = d+\omega_{\Spin}+A_{\Gamma},\label{eq:connection}%
\end{equation}
where $\omega_{\Spin}$ denotes the Spin connection and $A_{\Gamma}$ the
connection for the relevant duality bundle with structure group $\Gamma$. In this case, it is natural to restrict to $\Gamma = G_U$, but in general $\Gamma$ could be different from $G_U$.

Indeed, nothing really requires us to restrict ourselves to $D$-dimensional spacetimes with a $\Spin$ structure.\footnote{F-theory \cite{Vafa:1996xn, Morrison:1996na, Morrison:1996pp} is a prototypical example of this sort, in which the base of an elliptically-fibered Calabi-Yau manifold typically may only have a $\Spin^c$ structure rather than a $\Spin$ structure.} Rather, one can impose a milder condition in which the Spin connection and duality bundle are correlated in a Spin-$\Gamma$ bundle with structure group:
\begin{equation}
\Spin\text{-}\Gamma: \quad \frac{\Spin \times \Gamma}{\mathbb{Z}_2} \,,
\end{equation}
where the $\mathbb{Z}_2$ is the diagonal of the $(-1)^F$ generator of $\Spin(D-1,1)$ (i.e., it characterizes the cover $\Spin / \mathbb{Z}_2 = \text{SO}$), and some putative $\mathbb{Z}_2$ subgroup of $\Gamma$. From this perspective, we need to determine the Spin- and Pin$^{+}$-lifts of $G_U$. Again, we emphasize that if we restrict to supersymmetric backgrounds with Spin structure, the duality group will appear to be just $G_U$. If, however, we entertain more general backgrounds, more care will be needed.\footnote{Another way to see the same issue is to work with a fixed $T^{d}$, but with tuned moduli. At these enhanced loci, additional symmetries act on fermionic degrees of freedom, leading to the same conclusion.}

Our aim in the remainder of this section will be to understand the Spin- and Pin$^{+}$-lifts of these duality groups, respectively denoted by $\widetilde{G_{U}}$ and $\widetilde{G_U}^{+}$ (see also Appendix \ref{app:Spin_Pin_Lifts}).

\subsection{Warmup: Lifting \texorpdfstring{$\SL(d,\mathbb{Z})$}{SL(d,Z)}}

Before examining the various extensions of the full U-duality groups in \eqref{tab:Summary_Bos_U-dualities}, our aim in this section will be to first motivate the general discussion with an analysis of Spin- and Pin$^{+}$-lifts of $\SL(d,\mathbb{Z})$ dualities. These are extensions of the duality groups specified by the large diffeomorphisms of an internal $T^{d}$.\footnote{This example is also studied in \cite{Pantev:2016nze}.} After this, we will proceed to the more intricate case of U-dualities.

We first study the Spin-lift of $\SL(d,\mathbb{Z})$, the group of large diffeomorphisms of a $T^d$. The motivation for this case is a gravitational theory with fermionic degrees of freedom compactified on a $T^d$. In the presence of a non-trivial $T^d$ fibration over the $D$-dimensional spacetime, the fermionic degrees of freedom will have correlated Spin and duality bundle transformations.

With this in mind, it is instructive to start with the continuous group $\SL(d,\mathbb{R})$, which has maximal compact subgroup $\SO(d)$. From this, we see that the first homotopy groups are given by
\begin{equation}
    \pi_1 \big(\SL(d,\mathbb{R})\big) = \pi_1\big(\SO(d)\big) =
    \begin{cases}
        \mathbb{Z} \quad \,\;\textrm{if}\,\,\,d=2\,,\\
        \mathbb{Z}_2 \quad \textrm{if}\,\,\,d\geq3\,.
    \end{cases}
\end{equation}
Thus, for $d\geq3$ the Spin-lift / double cover of $\SL(d,\mathbb{R})$ is also the universal cover, and is given by a central extension of the form
\begin{equation}
    1 \rightarrow \mathbb{Z}_2 \rightarrow \widetilde{\SL}(d,\mathbb{R}) \rightarrow \SL(d,\mathbb{R}) \rightarrow 1 \label{eq:SL_Spin_Lift}.
\end{equation}
The main exception is $\SL(2,\mathbb{R})$, whose universal cover is a $\mathbb{Z}$-fold central extension. In this case, the metaplectic group $\Mp(2,\mathbb{R})$ is the only non-trivial double cover of $\SL(2,\mathbb{R})$. This singles out this extension as the appropriate one to act on fermions, which transform in $\Spin\text{-}\widetilde{\mathrm{SL}}$ bundles. From this, we can get the desired Spin-lift / central extension of the discrete group $\SL(d,\mathbb{Z})$ via a pullback from the extension \eqref{eq:SL_Spin_Lift}.

We can generalize even further by allowing for reflections, as necessary when considering the full Pin$^+$-lift of the duality group. In particular, we include a new generator $R$ in $\GL(d,\mathbb{Z}) \setminus \SL(d,\mathbb{Z})$. For example, one such generator is
\begin{equation}
    R = \textrm{diag}(-1,1,\dots,1) \in \GL(d,\mathbb{Z)} \setminus \SL(d,\mathbb{Z}).
\end{equation}
The element $R$ acts on the generators of $\SL(d,\mathbb{Z})$ via conjugation, and all other reflections are also related to $R$ by conjugation. Thus, the inclusion of this reflection enlarges the duality group to
\begin{equation}
    \GL(d,\mathbb{Z}) = \SL(d,\mathbb{Z}) \rtimes \mathbb{Z}_2^R.
\end{equation}
The full Pin$^+$-lift is then given by a $\mathbb{Z}_2$ central extension of the form (in the obvious notation):
\begin{equation}
    1 \rightarrow \mathbb{Z}_2 \rightarrow \widetilde{G}^+ \rightarrow \SL(d,\mathbb{Z}) \rtimes \mathbb{Z}_2^R \rightarrow 1,
\end{equation}
where $\widetilde{R} \in \widetilde{G}^+$ (the lifted reflection generator) satisfies $\widetilde{R}^2 = 1$, as demanded by the Pin$^+$ structure.\footnote{A Pin$^-$-lift would instead require $\widetilde{R}^2 = (-1)^F$.}

With this in hand, we now proceed to consider the more involved case of U-dualities in maximal supergravity theories.

\subsection{Spin- and \texorpdfstring{Pin$^+$}{Pin+}-Lifts of U-Dualities}

In this section we turn to the Spin- and Pin$^{+}$-lifts of U-dualities.

\subsubsection{Spin-Lifts}

Consider next the Spin-lift of the full U-duality group $G_U$, which contains $\SL(d,\mathbb{Z})$ as a subgroup. We denote this by $\widetilde{G_U}$ \cite{Pantev:2016nze}. This is once again defined via a central extension of the form
\begin{equation}
    1 \rightarrow \mathbb{Z}_2 \rightarrow \widetilde{G_U} \rightarrow G_U \rightarrow 1 \,. \label{eq:Spin_Extension}
\end{equation}
In particular, we find that $\pi_1 \big(G_U(\mathbb{R})\big) = \pi_1(K_U) = \mathbb{Z}_2$ for $3 \leq D \leq 7$, where $G_U(\mathbb{R})$ is the continuous bosonic U-duality group and $K_U\subset G_U(\mathbb{R})$ is the maximal compact subgroup. For $D = 8,9$ the situation is a bit more delicate (since there is a torsion free factor), but one can still identify a universal $\mathbb{Z}_2$ extension.

As discussed in Appendix \ref{App:Explicit_Spin_Lifts}, the Spin-lift for the discrete U-duality group $G_U$ can be defined by pulling back from that for $G_U(\mathbb{R})$, much as was the case for $\SL(d,\mathbb{Z})$.\footnote{This argument essentially follows the discussion in \cite{Pantev:2016nze}.} Once again, the main exception to this procedure is the group $\SL(2,\mathbb{Z})$, as the universal cover is infinite-sheeted. The metaplectic group $\Mp(2,\mathbb{Z})$ is still the unique non-trivial $\mathbb{Z}_2$ extension; it just arises from the unique non-trivial double cover of $\SL(2, \mathbb{R})$, which is not the universal cover.

\subsubsection{\texorpdfstring{Pin$^+$}{Pin+}-Lifts}

This story can be generalized even further by also allowing for compactifications of M-theory on non-orientable manifolds. In particular, 11D M-theory is well defined on manifolds with Pin$^+$ structure \cite{Witten:1996md, Freed:2019sco}, which in turn suggests a further lift of $G_U$ to incorporate orientation reversing transformations such as reflections.\footnote{In this work we are focusing on Spin-twisted duality bundles, so the overall Pin$^+$ structure of M-theory means that we must in turn restrict to the Pin$^+$-lift of $G_{U}$. If one considers a more general twisting by tangential structures (as can happen in some 9D supergravity backgrounds) then one might also entertain broader possible lifts. We leave the exploration of this possibility to future work.}

The lift to include reflections is given by a short exact sequence of the form
\begin{equation}
    1 \rightarrow G_U \rightarrow G_U^{(R)} \rightarrow \mathbb{Z}_2^R \rightarrow 1 \,, \label{eq:Reflection_Extension}
\end{equation}
which describes a semi-direct product
\begin{equation}
    G^{(R)}_U = G_U \rtimes \mathbb{Z}_2^R \,.
\end{equation}
The generator $R$ of the $\mathbb{Z}_2^R$ in this extension descends from a reflection element in the disconnected component of $G^{(R)}_U(\mathbb{R})$.\footnote{Once again, we only need to include a single new generator for the reflections, as all of the reflections are related by conjugation.} In practice, we determine $R$ as a matrix in the disconnected component of the maximal compact subgroup $K^{(R)}_U \subset G^{(R)}_U(\mathbb{R})$. This in turn defines a short exact sequence of the form
\begin{equation}
    1 \rightarrow K_U \rightarrow K_U^{(R)} \rightarrow \mathbb{Z}_2^R \rightarrow 1 \,.
\end{equation} 
The relevant maximal compact subgroups are summarized as
\begin{equation}
\centering
\begin{tabular}{|c|c|}
\hline
\textbf{$D$} & \makecell{\textbf{Maximal Compact}\\\textbf{Subgroup $K_U$}} \\
\hline
$9$ & $\SO(2)$  \\
$8$ & $\SO(3) \times \SO(2)$ \\
$7$ & $\SO(5)$ \\
$6$ & $(\Spin(5) \times \Spin(5))/\mathbb{Z}_2^{\textrm{diag}}$ \\
$5$ & $\USp(8)/\mathbb{Z}_2$ \\
$4$ & $\SU(8)/\mathbb{Z}_2$ \\
$3$ & $\Spin(16)/\mathbb{Z}_2$ \\
\hline
\end{tabular}
\label{tab:Maximal_Compact_Subgroups}\,.
\end{equation}

Finally, the full Pin$^+$-lift of interest is given by combining the two extensions \eqref{eq:Spin_Extension} and \eqref{eq:Reflection_Extension} into an extension of the form
\begin{equation}
    1 \rightarrow \mathbb{Z}_2 \rightarrow \widetilde{G_U}^+ \rightarrow G_U \rtimes \mathbb{Z}_2^R \rightarrow 1, \label{eq:Pin+_Extension}
\end{equation}
where we demand that the lifted reflection generator $\widetilde{R}$ satisfies $\widetilde{R}^2 = 1$. See Appendix \ref{App:Explicit_Pin_Lifts} for an explicit construction of the Pin$^+$-lifts in each dimension.\footnote{For $D \leq 8$ one can also define a reflection operator that squares to $(-1)^F$ by combining several reflections of the internal torus: see~\cite[\S 3.1]{Montero:2020icj}.}

As a final comment, we note that we have now fixed that the fermions in the $D$-dimensional effective theory transform as sections of a
\begin{equation}
    \frac{\Spin \times \widetilde{G_U}^+}{\mathbb{Z}_2}
\end{equation}
bundle over spacetime. With this in hand, we now move on to computing the relevant bordism groups, and analyze the predictions made by the Cobordism Conjecture in this setting.

\section{Bordisms and Branes} \label{sec:BORDISM}

In the previous section we discussed the Spin- and Pin$^{+}$-lifts of the bosonic U-duality groups. Our aim in this section will be to use the Swampland Cobordism Conjecture \cite{McNamara:2019rup} to predict new objects in the low energy effective theory.

The general setup we consider involves a $D$-dimensional effective field
theory which enjoys a duality symmetry $\Gamma$. We have in mind situations
where $\Gamma$ is realized as a $\mathbb{Z}_{2}$ extension, as appropriate for a Spin- and Pin$^{+}$-lift, which we
characterize by the short exact sequence:%
\begin{equation}
1\rightarrow \mathbb{Z}_{2}\rightarrow \Gamma \rightarrow \Gamma / \mathbb{Z}_2 \rightarrow
1 \,.\label{eq:GammaPrime}%
\end{equation}
In general, there can be a non-trivial correlation between the Spin structure of spacetime and the
duality bundle, and so we refer to Spin-$\Gamma$ twisted bundles as those
which mix the tangential and internal symmetries:\footnote{See, e.g., \cite{Debray:2023yrs} for a more mathematical definition.}%
\begin{equation}
\Spin\text{-} \Gamma \equiv\frac{\Spin\times \Gamma}{\mathbb{Z}_{2}} \,,
\label{eq:SpinGamma}%
\end{equation}
where the $\mathbb{Z}_2$ embeds as $(-1)^F$ in Spin and the image of the $\mathbb{Z}_2$ in the central extension \eqref{eq:GammaPrime} in the duality group $\Gamma$.

The Swampland Cobordism Conjecture \cite{McNamara:2019rup} asserts that when
the effective field theory has a non-trivial bordism group $\Omega_{k}^{\mathcal{G}}$ (where we consider bordism of manifolds with respect to whatever structure $\mathcal{G}$ is needed to define the theory), one
must enrich the low energy effective field theory by additional dynamical
objects of codimension $(k+1)$. One way to understand this condition is that
the presence of a non-trivial bordism group generator implies the existence of
a $p$-form symmetry $G^{(p)}=\textrm{Hom}(\Omega_{k}^{\mathcal{G}},\U(1))$, where $p=D-(k+1)$;
namely, one can construct an extended object filling $p$ spacetime dimensions.
Following the formulation of generalized global symmetries given in
\cite{Gaiotto:2014kfa}, there are corresponding topological symmetry operators
which link with these defects. In particular, they fill out $q=k$ dimensions
(so that $p+q=D-1$).

The general lore from quantum gravity is that all such putative global symmetries are actually broken (i.e., absent) or gauged. Let us first discuss breaking and then turn to gauging.
The topological linking can be destroyed if additional dynamical
states are added to the spectrum; namely, if our original defect can terminate
on an object which fills $p-1$ spatial dimensions, as well as time. This is
then a codimension-($k+1$) object in the theory. Observe that in terms of the
field content of the original effective field theory, this configuration is
necessarily singular, i.e., there is no deformation to a smooth
configuration. Supplementing the theory by these additional dynamical states
allows any putative topological linking to be destroyed, thus removing the
candidate global symmetry.

Gauging a finite $p$-form symmetry turns out to be somewhat subtle, and winds up introducing a magnetic dual global higher-form symmetry \cite{Gaiotto:2014kfa}. To get rid of these candidate symmetries we need to introduce additional degrees of freedom anyway. As such, we take the most conservative interpretation and include the expected defects right from the start.

Now, in the case at hand where we have fermions coupled to both the Spin
connection and a duality bundle, the relevant connection is of the schematic
form described in line (\ref{eq:connection}), which we reproduce here:
\begin{equation}
\mathcal{D} = d+\omega_{\Spin}+A_{\Gamma},
\end{equation}
where $\omega_{\Spin}$ denotes the Spin connection and $A_{\Gamma}$ the
connection for the relevant duality bundle, and there is a gauging by a diagonal $\mathbb{Z}_2$. This has the consequence that transition functions can be interpreted in multiple ways. For example, consider the transition function $(\theta, g) \in \Spin \times \Gamma$; due to the $\mathbb{Z}_2$ quotient, this is identified with
\begin{equation}
(\theta, g) \simeq ((-1)^F \theta, F g) \,,
\label{eq:twisttransfunc}
\end{equation}
where we denote the image of $1\in \mathbb{Z}_2$ in $\Gamma$, determined from \eqref{eq:GammaPrime}, by $F$. This is precisely the reason why the twisted theory can be formulated on non-Spin manifolds, because the cocycle condition on triple overlaps involves the equivalence class of $(\theta, g)$ and not $\theta$ alone. Thus, the obstruction to Spin structure, captured by the cocycle condition for $\theta$, is compensated by a non-trivial duality bundle, captured by the transition functions $g$. But even in the case the underlying manifold $X$ allows for a Spin structure, the identification in \eqref{eq:twisttransfunc} has the important consequence that the choice of Spin structure, classified by an element in $H^1(X;\mathbb{Z}_2)$, can be reinterpreted as a choice of duality bundle. The most important situation for us is the circle $S^1$, which we will discuss in more detail next.

The circle allows for two different Spin structures, because $H^1(S^1;\mathbb{Z}_2) = \mathbb{Z}_2$. The two Spin structures are characterized by periodic, $S^1_{+}$, or anti-periodic, $S^1_-$, boundary conditions for fermions. The $\Gamma$ bundles are characterized by $H^1(S^1; \Gamma)$,\footnote{If $G$ is a non-Abelian group, the standard definition of $H^k(X; A)$ with coefficients in an Abelian group $A$ can be extended to $G$ in a sensible way for $k\le 1$, and like in the Abelian case, $H^1(X; G)$ classifies principal $G$-bundles through their monodromy around non-trivial $1$-cycles.} which for discrete $\Gamma$ can be identified with $\Gamma$ itself. This means that $\Gamma$-bundles are classified by a transition function / monodromy $g$ when going around the circle. Denoting a circle with boundary conditions $\pm$ and monodromy $g$ as $(S^1_{\pm})_{g}$, the gauging of the diagonal $\mathbb{Z}_2$ leads to the identification
\begin{equation}
(S^1_{\pm})_{g} \simeq (S^1_{\mp})_{F  g} \,,
\label{eq:twistident}
\end{equation}
i.e., one can trade a $(-1)^F$ for $F$. This leads to the reduction of inequivalent manifolds, which will be important for the determination of the bordism groups below.

We are interested in calculating the bordism group $\Omega_{1}^{\Spin\text{-}\Gamma}$(pt). We present a direct calculation of this in Appendix
\ref{app:Adams}, where we establish that, when $D \leq 7$ and $ \Gamma
=\widetilde{G_{U}}$, this bordism group is trivial. Meanwhile, for $D\leq7$ and
$\Gamma=\widetilde{G_{U}}^{+}$, we instead have a single $\mathbb{Z}_{2}$ factor associated with monodromy by a reflection in the internal
toroidal directions. The cases $D = 8,9$ have additional (supersymmetry preserving) bordism generators.

Our aim here will be to establish the same result in more physical terms.
The end result of our analysis is that for a non-trivial $\Spin\text{-} \Gamma$ bundle, the bordism group is:\footnote{In general, bordism groups may be defined for manifolds equipped with a map to a fixed topological space $X$, and are denoted $\Omega_k^{\mathrm{Spin}\text{-}\Gamma}(X)$. An element of $\Omega_k^{\mathrm{Spin}\text{-}\Gamma}(X)$ is represented by a closed $k$-dimensional Spin-$\Gamma$ manifold $M$ together with a continuous map $f\colon M\to X$. When $X=\mathrm{pt}$, every such map is trivial, and one recovers the bordism group of closed Spin-$\Gamma$ manifolds without any additional structure. In this case one often omits the argument and simply writes $\Omega_k^{\mathrm{Spin}\text{-}\Gamma}$ instead of $\Omega_k^{\mathrm{Spin}\text{-}\Gamma}(\mathrm{pt})$. From a physical viewpoint, the argument $X$ encodes additional background fields via the map $M\to X$, and taking $X=\mathrm{pt}$ means that no such additional background fields are present beyond the Spin-$\Gamma$ structure already specified.}
\begin{equation}
\Omega_{1}^{\Spin\text{-} \Gamma} (\text{pt}) = \mathrm{Ab}[ \Gamma] \,.
\end{equation}

\subsection{Spin-Lifts, \texorpdfstring{Pin$^+$}{Pin+}-Lifts, and Bordisms} 

In this section we turn to the Spin- and Pin$^{+}$-Lifts of U-duality groups, and the corresponding first bordism groups associated with Spin-twisted duality bundles.

\subsubsection{Spin-Lift and Bordisms}
\label{ss:spin_lift_bordism}

Consider first the Spin-lift of the U-duality groups, $\widetilde{G_U}$. As discussed above, the only manifolds we need to consider are circles $(S^1_{\pm})_{g}$ with $\widetilde{G_U}$ monodromy $g$. Using the identification \eqref{eq:twistident}, we always fix the fermionic boundary conditions to be bounding. This can modify the duality bundle. More concretely, for non-bounding boundary conditions we rewrite
\begin{equation}
\label{switch_periodicity}
    (S^1_+)_{g} \simeq (S^1_-)_{F g} \,.
\end{equation}
However, since $\widetilde{G_U}$ is perfect (see Appendix \ref{app:LHS}) for $D \leq 7$, the element $F g$ is part of the commutator subgroup of $\widetilde{G_U}$. Similar to the results in \cite{McNamara:2021cuo, Debray:2023yrs, Ruiz:2024jiz, Braeger:2025kra}, this provides a gravitational soliton configuration with a well-defined twisted Spin structure that bounds the one-dimensional manifold (see Figure \ref{fig:Spintriviality}).
\begin{figure}
    \centering
    \includegraphics[width = 0.9 \textwidth]{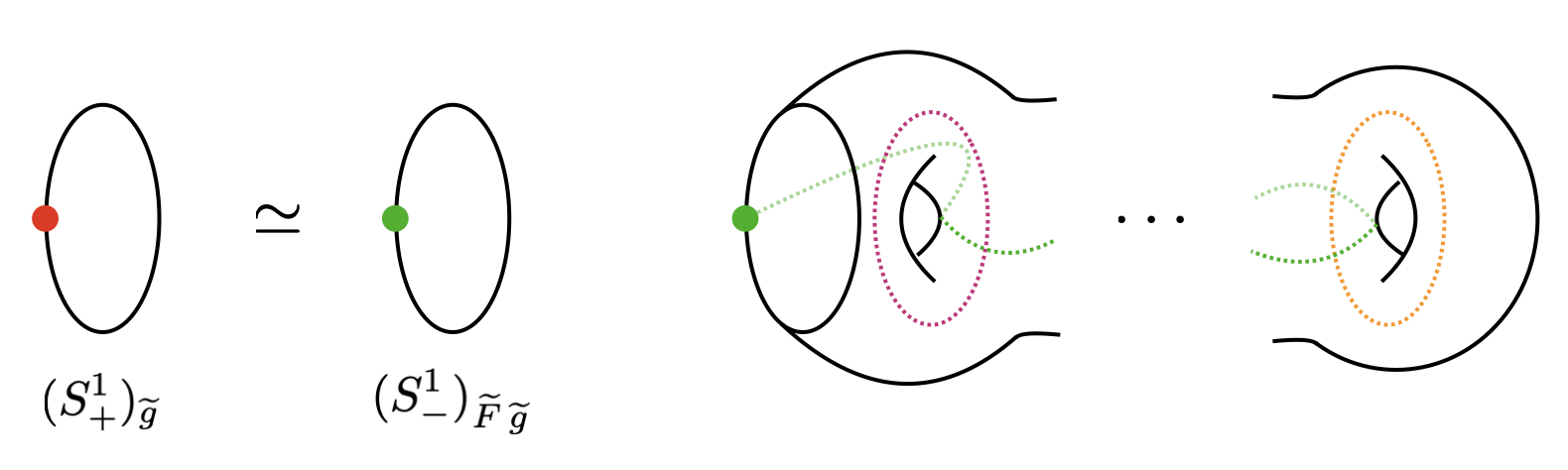}
    \caption{Bounding of every one-dimensional manifold with Spin-$\widetilde{G_U}$ structure (the duality bundle is depicted via transition functions on codimension-one sub-manifolds).}
    \label{fig:Spintriviality}
\end{figure}
We therefore see that, for $D \leq 7$,
\begin{equation}
    \Omega^{\Spin\text{-}\widetilde{G_U}}_1 (\text{pt}) = 0 \,.
\end{equation}
This is also proven by a spectral sequence argument in Appendix \ref{app:Adams}. Since the bordism group vanishes there is no need to include any codimension-two Spin-$\widetilde{G_{U}}$ defects in the theory to break global symmetries.

The story is different for duality groups with non-trivial Abelianization, i.e., for $D = 8,9$, since $G_U$ has a simple $\SL(2,\mathbb{Z})$ factor. There, the Spin-lift can lead to a non-trivial extension. For example, as found in \cite{Debray:2023yrs}:
\begin{equation}
    \Omega^{\Spin}_1 \big( B\SL(2,\mathbb{Z}) \big) = \mathbb{Z}_2 \oplus \mathbb{Z}_{12} \rightarrow \mathbb{Z}_{24} = \Omega^{\Spin\text{-}\Mp(2,\mathbb{Z})}_1 (\pt) \,.
\end{equation}
Again, we can understand this pictorially by translating the Spin structure on the spacetime circle to a shift in the transition function of the duality bundle. This reasoning is compatible with the relation
\begin{equation}
    \Omega^{\Spin\text{-} \Gamma}_1(\pt) = \text{Ab} [\Gamma ] \,,
\end{equation}
for more general $\Gamma$, which we prove in Appendix \ref{app:Adams} when $\Gamma$ is one of the U-duality groups.

\subsubsection{\texorpdfstring{Pin$^+$}{Pin+}-Lift and Bordisms}
\label{ss:pinp_lift_bordism}

Next, we turn to the Pin$^+$-lift of the duality group $G_U$ after the inclusion of a reflection element $R$, which can be understood as extending $G_U$ to the semi-direct product $G_U \rtimes \mathbb{Z}_2^R$. The fact that we do a Pin$^+$-lift means that $R$ lifts to an element in $\widetilde{G_U}^+$, which we denote by $\widetilde{R}$, that still squares to the identity
\begin{equation}
    \widetilde{R}^2 = 1 \in \widetilde{G_U}^+ \,.
\end{equation}
As before, the $\mathbb{Z}_2$ of the extension defining the lift is identified with the $(-1)^F$ coming from Spin$(D-1,1)$, and we can translate the choice of Spin structure into a duality bundle.

The different backgrounds we need to consider are given by circles with a duality transition function. With the exact same argument as above, using that each element in $G_U$ as well as $F$ can be written in terms of a commutator, there are bounding configurations of the type depicted in Figure \ref{fig:Spintriviality}. The only generator left to discuss is $(S^1_-)_{\widetilde{R}}$, i.e., the circle with a transition function given by reflection. However, the reflection element can never be a commutator and thus cannot be bounded by manifolds of the type above.

For that let us recall that the Pin$^+$-lift can be described by the short exact sequence
\begin{equation}
    1 \rightarrow \mathbb{Z}_2 \rightarrow \widetilde{G_U}^+ \rightarrow G_U \rtimes \mathbb{Z}_2^R \rightarrow 1 \,.
\end{equation}
Now if there was a way to write $\widetilde{R}$ as a (product of) commutator(s), we could map this to $G_U \rtimes \mathbb{Z}_2^R$, where the equation should still hold and produce $R$ as a commutator. Note, however that $G_{U}^{(R)}(\mathbb{R})$ has two disconnected components, and under the ``determinant map''  $\mathrm{det}:G_{U}^{(R)}(\mathbb{Z}) \rightarrow \mathbb{Z}_2^{(R)}$,\footnote{For any representation it is just a determinant on the corresponding linear maps.} the image of $G_U$ is $+1$ and the element $R$ maps to $-1$. Then, it is clear that $R$ cannot be a commutator, as it always contains an even number of elements with determinant $-1$ and hence can only produce elements with positive determinant. The same is true for $\widetilde{R}$ in $\widetilde{G_U}^+$.\footnote{Though not relevant for our present physical purposes, the same holds true for the Pin$^{-}$ lift.}

One might also ask whether there might be other manifolds that can bound the circle with $\widetilde{R}$ transition function which are not included in the class above. If that was true, one could glue two of them along their common boundary, given by $(S^1_-)_{\widetilde{R}}$, to obtain a compact, orientable, and smooth two-dimensional manifold. These are simply the Riemann surfaces, which are captured by the argument above. Any other such manifolds cannot exist.

With the arguments above, for $G_U$ perfect (i.e., for $D \leq 7$), we find that,
\begin{equation}
    \Omega^{\Spin\text{-}\widetilde{G_U}^+}_1 (\pt) = \mathbb{Z}_2 \,\,\, \text{for} \,\,\, 3 \leq D \leq 7,
\end{equation}
with the generator given by $(S^1_-)_{\widetilde{R}}$. This is in accord with the general discussion given above, and the result:
\begin{equation}
    \Omega^{\Spin\text{-} \Gamma}_1 (\pt) = \text{Ab} [\Gamma] \,,
\end{equation}
with proof in Appendix \ref{app:Adams} in the case that $\Gamma$ is one of the U-duality groups. This includes $D = 8,9$ for which one has
\begin{equation}
 \Omega^{\Spin\text{-}\widetilde{G_U}^+}_1 (\pt) = \text{Ab} \big[ \widetilde{G_U}^+ \big] = \mathbb{Z}_2 \oplus \mathbb{Z}_2 \,\,\, \text{for} \,\,\, D = 8,9,
\end{equation}
with one $\mathbb{Z}_2$ associated to a circle with reflection and the other $\mathbb{Z}_2$ to a circle with non-trivial $\Mp(2,\mathbb{Z})$ monodromy given by a Spin-lift of the $S$ generator for $\SL(2,\mathbb{Z})$.

We thus conclude that in the original $D$-dimensional effective theory, we get
a non-trivial codimension-two defect, as obtained from the Pin$^{+}$-lift of
reflections on the internal torus directions. We refer to this as a
\textquotedblleft reflection brane\textquotedblright\ since it arises from an
internal reflection on the torus. Note also that nothing singles out a
particular direction of reflection. Indeed, conjugating by the internal
$\SL(d,\mathbb{Z})$ transformations, or even $G_U$ transformations, this reflection can instead act on any of the internal $T^{d}$ torus coordinates and involve T-dualities. The main point is that under monodromy around the defect, the orientation of $T^{d}$ reverses:%
\begin{equation}
T^{d}\rightarrow\overline{T^{d}}\text{.}%
\end{equation}

\subsection{Extra Branes at \texorpdfstring{$D = 8,9$}{D=8,9}}

Summarizing the above discussion, for $3 \leq D \leq 7$, the bordism group $\Omega_{1}^{\Spin\text{-}\widetilde{G_U}^{+}}(\pt) = \mathbb{Z}_2$, and this predicts the existence of a reflection brane which trivializes the corresponding bordism class in $\Omega_{1}^{\mathrm{QG}}$. For $D = 8,9$, the Abelianization of $\widetilde{G_U}^+$ is actually $\mathbb{Z}_2 \oplus \mathbb{Z}_2$, so there is an additional generator to contend with. This is essentially because in both of these cases, $G_U$ contains a non-trivial $\SL(2,\mathbb{Z})$ factor, and $\mathrm{Ab}[G_{U}] = \mathbb{Z}_{12}$. The other brane predicted by these generators is essentially the same one already discussed in \cite{Debray:2023yrs}: it is a supersymmetric codimension-two object.

For example, in $D = 9$, this is characterized by a $T^{2}$ fibration over $\mathbb{C}$ in which the complex structure is fixed to $\tau = i$. There are different 6-brane configurations which realize this. For example, we can use a Kodaira fiber of type III$^{\ast}$ as well as a Kodaira fiber of type III. Collapsing this elliptic fiber to zero size would result in an $\mathfrak{e}_7$ (for type III$^{\ast}$) and $\mathfrak{su}_2$ gauge symmetry (for type III). Indeed, essentially the same configuration was considered in the F-theory backgrounds of IIBordia in \cite{Debray:2023yrs}.

In the case of $D = 8$, we can realize the relevant generator by working with type IIB string theory compactified on a $T^{2}_{\mathrm{spatial}}$. Allowing the elliptic fibration of F-theory to have precisely the form indicated above amounts to compactifying the relevant $\tau = i$ 7-brane configuration on a further $T^{2}_{\mathrm{spatial}}$. This is essentially the same strategy also used in \cite{Braeger:2025kra} to partially geometrize the U-duality groups in M- / F-theory.

\section{Properties of Reflection Branes} \label{sec:REFLECTIONBRANES}

In the previous section we used the Swampland Cobordism Conjecture to argue
for the existence of codimension-two reflection branes, as captured by a $\mathbb{Z}_{2}$ factor of $\mathrm{Ab}[ \widetilde{G_{U}}^{+} ]$. In this section we
establish some basic properties of these objects. Some aspects of this analysis
amount to generalizations of what was found in \cite{Dierigl:2022reg} for the
R7-branes of type IIB / F-theory (see also \cite{Heckman:2025wqd}). Indeed, one can view our reflection branes
as descending from R7-branes wrapped on internal cycles. We begin by
establishing that these objects break supersymmetry, but are nevertheless stable.\footnote{This is to be contrasted with the case of unstable non-supersymmetric orbifolds, see e.g., \cite{Morrison:2004fr, Adams:2001sv, Harvey:2001wm, Dabholkar:2001wn, Martinec:2001cf, Vafa:2001ra, Braeger:2024jcj, Braeger:2025rov}. Stability in the present 
case follows from similar considerations to those presented in \cite{Kaidi:2024cbx} (see also \cite{Heckman:2025wqd}).} We then turn to an
analysis of the BPS\ objects which can terminate on our reflection branes.
Following this, we consider some preliminary aspects of their dynamics,
primarily focusing on properties well constrained by topological
considerations. Along these lines, we study the braiding of branes constructed
from different reflections, as well as the class of bound states these objects form.

At a general level, we observe that the reflection branes considered here can be viewed as arising from the IIA reflection 7-brane of type IIA string theory (see \cite{Heckman:2025wqd}). The M-theory lift of this configuration is simply a cone over a Klein bottle, in which the Klein bottle is viewed as a circle fibration, where the fiber undergoes orientation reversal in winding around the base circle.\footnote{Indeed, the bordism class of the Klein bottle generates $\Omega_{2}^{\Pin^{+}}(\mathrm{pt})$~\cite[Proposition 3.9]{KT90}, and the corresponding defect predicted by the Swampland Cobordism Conjecture \cite{McNamara:2019rup} is the M-theory lift of the IIA R7-brane \cite{Heckman:2025wqd}.} As such, many of the qualitative properties of a single reflection brane follow directly from compactification of such objects. With this in mind, many of the salient properties, such as supersymmetry breaking, stability, as well as the spectrum of objects which can terminate on these branes directly follow from toroidal compactification of the IIA R7-brane.

\subsection{SUSY Breaking and Stability}

In this section we argue that the reflection brane completely breaks supersymmetry in the $D$-dimensional effective supergravity theory.

As a warmup, we first consider type IIA string theory with a $(-1)^{F_L}$ R7-brane and some supersymmetric D$p$-brane probes. Passing one such brane through the branch cut generates an anti-D$p$-brane, so the combination of branes and anti-branes (and R7-brane) breaks all supersymmetries. Since we can compactify this configuration on a $T^d$, we conclude that the related reflection branes probed by BPS branes also break supersymmetry. See Figure \ref{fig:Dp_Anti-Dp_Sytem} for a depiction of this configuration.

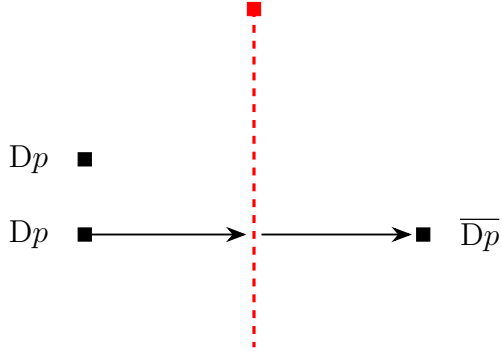
\begin{figure}
    \centering
    \begin{tikzpicture}
	\begin{pgfonlayer}{nodelayer}
		\node [draw=red, fill=red, style=none, minimum size=5pt] (0) at (0, 0) {};
		\node [style=none] (1) at (0, -4.5) {};
		\node [draw=black, fill=black, style=none, minimum size=5pt] (2) at (-2.25, -2) {};
		\node [draw=black, fill=black, style=none, minimum size=5pt] (3) at (-2.25, -3) {};
		\node [draw=black, fill=black, style=none, minimum size=5pt] (4) at (2.25, -3) {};
		\node [style=none] (5) at (-0.1, -3) {};
        \node [style=none] (6) at (0.1, -3) {};
		\node [style=none] (7) at (-3, -2) {D$p$};
		\node [style=none] (8) at (-3, -3) {D$p$};
		\node [style=none] (9) at (3, -3) {$\overline{\textrm{D}p}$};
        \node [style=none] (10) at (2.1, -3) {};
	\end{pgfonlayer}
	\begin{pgfonlayer}{edgelayer}
		\draw [color=red, style=dashed, very thick] (0.center) to (1.center);
		\draw [style=ArrowLineRight] (3.center) to (5.center);
		\draw [style=ArrowLineRight] (6.center) to (10.center);
	\end{pgfonlayer}
\end{tikzpicture}
    \caption{Depicition of a pair of BPS D$p$-branes near a IIA R7-brane. Passing one D$p$-brane through the branch cut of the R7-brane (red) becomes a $\overline{\textrm{D}p}$-brane when the $(p+1)$-form potential it is coupled to picks up a minus sign. Under toroidal compactification, similar considerations hold for all reflection branes and probes by BPS objects.}
    \label{fig:Dp_Anti-Dp_Sytem}
\end{figure}

With this in hand, we will now show that the reflection brane itself breaks supersymmetry in the $D$-dimensional effective theory. To that end, recall that a background is only supersymmetric if there is a non-trivial solution to the Killing spinor equation:
\begin{equation}
    \nabla_\mu Q = 0. \label{eq:Killing_spinor}
\end{equation}
Here $\nabla_\mu$ depends on background fields and $Q$ is a complex supercharge in the $D$-dimensional theory. We claim that upon including the reflection brane there are no solutions to \eqref{eq:Killing_spinor}.\footnote{This was previously shown for the R7-brane in \cite{Dierigl:2022reg}.}

We can see that supersymmetry is broken by studying the monodromy action on any candidate solution to \eqref{eq:Killing_spinor}. The reflection branes are codimension-two defects in the $D$-dimensional theory. As such, it is instructive to split the $D$-dimensional spacetime as $\mathbb{R}^{D-3,1} \times \mathbb{R}^2$. Then, the reflection brane breaks the Lorentz algebra in $D$-dimensions to
\begin{equation}
    \mathfrak{so}(D-1,1) \rightarrow \mathfrak{so}(D-3,1) \times \mathfrak{so}(2), \label{eq:Lorentz_Algebra_Breaks}
\end{equation}
where $\mathfrak{so}(D-3,1)$ is the Lorentz algebra in ($D-2$)-dimensions along the brane worldvolume and $\mathfrak{so}(2)$ are transformations in $\mathbb{R}^2$, which can be understood as rotations in the plane perpendicular to the reflection brane.

Under rotations by an angle $\theta$ associated to the $\mathfrak{so}(2)$ factor of \eqref{eq:Lorentz_Algebra_Breaks}, the supercharge $Q$ transforms as\footnote{Note that this is the standard transformation of spin-1/2 operators under rotations.}
\begin{equation}
    Q \mapsto e^{i\theta/2}Q. \label{eq:Supercharge_Transformation}
\end{equation}
On the other hand, $Q$ undergoes a monodromy under a full rotation around the reflection brane:
\begin{equation}
    Q \mapsto \omega\overline{Q}, \label{eq:Condition}
\end{equation}
where $\omega$ is a phase related to the net conical deficit angle obtained by circling around the reflection brane. Importantly, there are no solutions to \eqref{eq:Supercharge_Transformation} that can satisfy the condition \eqref{eq:Condition}, so the reflection brane must break supersymmetry in the $D$-dimensional theory. This can be seen by noting that any parity reversing operation, such as reflections, has determinant $-1$. On the other hand, any transformation under $\mathfrak{so}(2)$ must have determinant $+1$.

\paragraph{Stability}

While the reflection branes are charged, and thus cannot disappear entirely, the fact that these objects are non-supersymmetric might at first suggest that they are unstable and will eventually expand into an energetically favorable configuration. However, the branes in question are only charged under a discrete group which does not embed in a continuous group. Indeed, if the brane was actually unstable, then there would have to be a smooth field configuration in the low energy supergravity theory for the brane to expand into. The obstructions to such a configuration are characterized by the SUGRA bordism groups $\Omega_\ast^{\textrm{SUGRA}}$ \cite{Kaidi:2024cbx}.\footnote{See \cite{Heckman:2025wqd} for the case of the R7-brane.} In other words, objects in $\Omega_\ast^{\textrm{SUGRA}}$ cannot be realized as a smooth configuration in the effective theory and thus must be stable against deformations to EFT configurations.\footnote{One could of course hypothesize some as yet unknown deformation to another singular configuration in the UV completion, but one might as well refer to this object as the same thing as the original one predicted by the Swampland Cobordism Conjecture.} Indeed, the reflection branes in this paper are exactly objects in such bordism groups descending from M-theory compactifications.

\subsection{Bubbles and Walls}

Much as in the case of the reflection 7-branes of type IIA and IIB, we also expect these
reflection branes to arise from codimension-one walls on
collapsing cylindrical configurations with a monodromy cut. 
Similar to reference \cite{Heckman:2025wqd}, we consider a IIA/IIB wall with an
$(-1)^{F_{L}}$ cut emanating out of the wall. We put IIB\ in the interior of
the cylindrical configuration, and IIA on the outside. This collapses to the
R7-brane of type IIA, which in turn lifts in M-theory to a cone over a Klein
bottle. We can take this same setup and compactify over an additional $T^{d}$.
As such, we see that there is a bubble-like configuration which collapses to
our reflection brane.\footnote{What about collapsing configurations involving multiple reflections? This is is considerably more
subtle. The issue is that, as of this writing, the corresponding wall which can
collapse to other reflection branes is not known (see \cite{Heckman:2025wqd}
for further discussion on this point). Because of this complication, we defer an
analysis of this case to future work.} See Figure \ref{fig:IIA_IIB_Wall} for a depiction.
\begin{figure}
    \centering
    \includegraphics[width=0.75\linewidth]{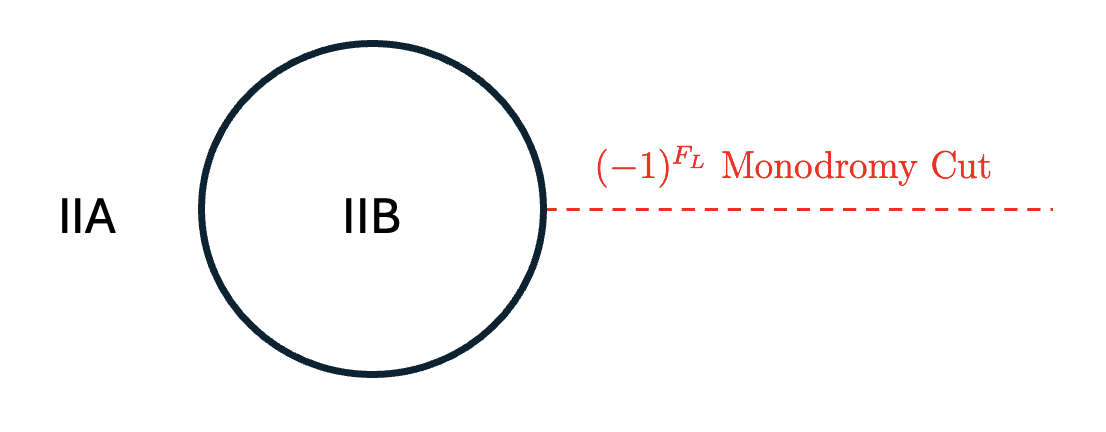}
    \caption{Top down view of the cylindrical configuration of the type IIA/IIB wall with topology $S^1 \times \mathbb{R}^{D-1} \times T^d$ with a $(-1)^{F_L}$ monodromy cut. IIB is on the inside of the wall while IIA is on the outside. The configuration collapses due to the tension of the wall, and the endpoint of the collapse is the reflection brane. This configuration lifts to a M- / F-theory wall.}
    \label{fig:IIA_IIB_Wall}
\end{figure}

\subsection{Lasso Configurations and Worldvolume Degrees of Freedom} \label{subsec:Lasso}

We now determine some physical properties of the reflection branes by probing them with known supersymetric branes \cite{Dierigl:2022reg}. The probe analysis relies on determining how the various supergravity $p$-form potentials transform under a reflection, i.e., whether the reflection brane acts via charge conjugation on the supersymmetric brane coupled to the $p$-form potential or not.

One of two things can happen: either the $p$-form potential picks up a minus sign under a reflection, in which case the reflection brane acts via charge conjugation on the probe supersymmetric brane, or the $p$-form potential is invariant under a reflection. In the former case, two copies of the probe supersymmetric brane, that extend out to infinity, can be joined by passing through the branch cut of the reflection brane. The pair of probe branes then combine at a junction (see e.g., \cite{Gaberdiel:1997ud, Gaberdiel:1998mv, Cvetic:2021sxm, Cvetic:2022uuu}), completing  a lasso around the reflection brane. The lasso then collapses to an energetically favored configuration, leaving an even number of probe supersymmetric branes that end on the reflection brane. See Figure \ref{fig:Lasso} for an example.

\begin{figure}[t!]
    \centering
    \begin{tikzpicture}
	\begin{pgfonlayer}{nodelayer}
		\node [draw=black, fill=black, style=none, minimum size=5pt] (0) at (-5, 1) {};
		\node [style=none] (1) at (-4, 0) {};
		\node [style=none] (2) at (-6, 0) {};
		\node [style=none] (3) at (-5, -1) {};
		\node [draw=red, fill=red, style=none, minimum size=5pt] (4) at (-5, 0) {};
		\node [style=none] (5) at (-5, -3) {};
		\node [style=none] (6) at (-5, 3.25) {};
		\node [style=none] (7) at (-5, -4.25) {(i)};
		\node [style=none] (8) at (-6.75, 0) {D$p$};
		\node [style=none] (9) at (-3.25, 0) {D$p$};
		\node [style=none] (10) at (-3.75, 2.25) {$2\times$D$p$};
		\node [draw=red, fill=red, style=none, minimum size=5pt] (11) at (3, 0) {};
		\node [style=none] (12) at (3, -3) {};
		\node [style=none] (13) at (3, -4.25) {(ii)};
		\node [style=none] (14) at (3, 3.25) {};
		\node [style=none] (15) at (4, 2.25) {$\#$D$p$};
		\node [style=none] (16) at (-1.75, 0) {};
		\node [style=none] (17) at (1.5, 0) {};
	\end{pgfonlayer}
	\begin{pgfonlayer}{edgelayer}
		\draw [bend left=45, looseness=1.50, thick] (2.center) to (0.center);
		\draw [bend left=45, looseness=1.50, very thick, style=ArrowLineRight] (3.center) to (2.center);
		\draw [bend right=45, looseness=1.50, very thick, style=ArrowLineRight] (3.center) to (1.center);
		\draw [bend right=45, looseness=1.50, thick] (1.center) to (0.center);
		\draw [color=red, style=dashed, very thick] (4.center) to (5.center);
		\draw [very thick, style=ArrowLineRight] (0.center) to (6.center);
		\draw [color=red, style=dashed, very thick] (11.center) to (12.center);
		\draw [very thick, style=ArrowLineRight] (11.center) to (14.center);
		\draw [style=ArrowLineRight] (16.center) to (17.center);
	\end{pgfonlayer}
\end{tikzpicture}
    \caption{Type IIA $(-1)^{F_L}$ R7-brane in the presence of BPS brane probes. (i) Two D$p$-branes (black lines) are joined by passing through the branch cut of the reflection brane (red). The branes combine at a junction (black square) and extend to infinity. (ii) If two D$p$-branes can be lassoed in this way, then this implies that any number of the brane (not just an even number) can terminate on the R7-brane. Wrapping these branes on directions of an internal $T^{d}$ results in similar configurations for reflection branes probed by BPS objects of the $D$-dimensional effective field theory.}
    \label{fig:Lasso}
\end{figure}
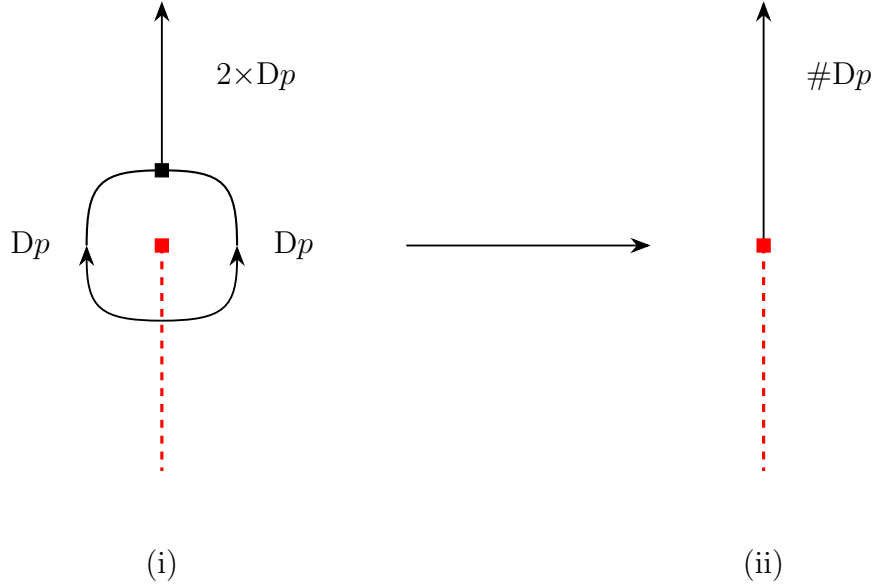

Note also that since the reflection branes arise as a collapsing cylindrical configuration separating F-theory and M-theory backgrounds, we can also ask about the fate of such branes as they pass from one side of the M- / F-theory wall (the lift of the IIA/IIB wall) to the other side (see Figure \ref{fig:IIA_IIB_Wall}). Following the discussion in \cite{Heckman:2025wqd}, this in turn means that once we establish that a brane-lasso configuration can terminate on a reflection brane, we can actually strengthen the conclusion to argue that a single supersymmetric brane (instead of a pair) can actually terminate on it.

By determining the various types of supersymmetric branes that end on the reflection brane, we are also able to partially uncover the worldvolume degrees of freedom of the reflection branes themselves. It is important to note that our procedure is only sensitive to branes that can be ``lassoed" in this way. Thus, the full worldvolume theory of the reflection branes may be more complicated. In the following, we review the R7-branes of type IIA and perform a detailed analysis for the reflection 6-brane, which is obtained by a circle compactification of the R7-brane. We then generalize the story to all the reflection branes, which are obtained via further circle compactifications.

\subsubsection{Reflection 7-Branes of Type IIA}

Our primary interest is in the reflection branes generated by $D \leq 9$ M-theory vacua. That being said, some of our considerations already follow from codimension-two defects of type IIA, namely the R7-brane associated with $(-1)^{F_L}$ monodromy (see \cite{Dierigl:2022reg, Heckman:2025wqd}). As such, we already know that all of the D$p$-branes go to anti-D$p$-branes under monodromy. We also know that there are lasso configurations where pairs of these branes can terminate on the reflection 7-brane. Furthermore, using the general arguments provided in \cite{Heckman:2025wqd}, where we construct the 7-brane from a collapsed IIA / IIB wall wrapped on a cylindrical configuration with an $(-1)^{F_L}$ monodromy cut, one can actually conclude that a single D$p$-brane can terminate on the R7-brane.\footnote{Another way to reach the same conclusion is to consider a pair of branes which end on the R7-brane, and to then let one of the branes move off to infinity. While this may not be energetically preferred, nothing obstructs this deformation at the level of realizing off-shell configurations.}

In the lift to M-theory, the D$0$-brane descends from KK momenta on the M-theory circle, and the D$6$-brane is its magnetic dual counterpart. Likewise, the F1-string descends from a wrapped M$2$-brane, and the NS$5$-brane and D$4$-brane are the magnetic dual counterparts.

Our aim in the remainder of this section will be to characterize the relevant supersymmetric objects which can terminate on the reflection brane of the $D$-dimensional effective field theory. Since the objects associated with KK momenta behave universally across all spacetime dimensions, we primarily focus on those degrees of freedom which couple to the M-theory $3$-form potential $C_3$ or its magnetic dual $6$-form potential $\widetilde{C}_{6}$.

\subsubsection{Reflection 6-Brane}

We begin by analyzing all of the possible lasso configurations involving the reflection 6-brane, which is the codimension-two defect in the 9D effective supergravity theory predicted by the Cobordism Conjecture. The possible supersymmetric branes / $p$-form potentials in the 9D theory descend from the M2-branes coupled to the M-theory 3-form $C_3$, as well as the magnetic dual M5-branes coupled to the M-theory 6-form $\widetilde{C}_6$. It is important to note that $C_3$ is a pseudo 3-form while $\widetilde{C}_6$ is a real 6-form.\footnote{We can see this by studying the 11D supergravity action, which should be invariant under parity transformations such as reflections. In particular, notice that the topological Chern-Simons term
\begin{equation}
\int C_3 \wedge G_4 \wedge G_4,
\end{equation}
where locally $G_4 = dC_3$, is only reflection invariant if we take $C_3$ to be a pseudo 3-form. 
Likewise, the kinetic term
\begin{equation}
    \int G_4 \wedge G_7,
\end{equation}
where locally $G_7 = d\widetilde{C}_6$, is only invariant under reflections if $G_7$, and consequently $\widetilde{C}_6$, is a real 6-form.}

With this in hand, and to be concrete, we first partition the internal $T^2$ of the 11D supergravity compactification as
\begin{equation}
    \mathbb{R}^{8,1} \times S^1_{(1)} \times S^1_{(2)},
\end{equation}
where the subscript denotes the two cycles of the $T^2$. Furthermore, and without loss of generality, let the reflection $R_1$ act on the first cycle $S^1_{(1)}$.

Let us begin by analyzing the descendants of $C_3$ that undergo a parity transformation under this setup. The possible candidates include 0-form, 1-form, 2-form, and 3-form potentials.

\paragraph{0-Form Potentials from $C_3$} There are no candidate 0-form potentials descending from $C_3$ in the 9D effective theory. This is because any candidate 0-form could only emerge if all of $C_3$ was compactified in the internal directions, which is not possible in this case.\footnote{Note that there can still be scalars in the $D$-dimensional theory that do not originate from $C_3$ and undergo non-trivial monodromy. For example, the complex structure modulus has monodromy $\tau \rightarrow - \overline{\tau}$.}

\paragraph{1-Form Potentials from $C_3$} Starting from $C_3$, there is a single candidate 1-form potential in the 9D effective theory that descends from compactifying $C_3$ on both internal directions. We denote the 1-form as $A_1^{[12]}$ where the superscript denotes the compactified directions. However, we see that
\begin{equation}
    A_1^{[12]} \xrightarrow{R_1} A_1^{[12]},
\end{equation}
as the 1-form picks up two minus signs under the reflection (one from $C_3$ and one from the reflection of the coordinate itself). Thus, no branes coupled to $A_1^{[12]}$ can terminate on the reflection 6-brane via lasso configurations.

\paragraph{2-Form Potentials from $C_3$} Starting from $C_3$, there are two candidate 2-form potentials in the 9D effective theory that descend from compactifying $C_3$ on either $S^1_{(1)}$ or $S^1_{(2)}$. We denote them as $B_2^{[1]}$ and $B_2^{[2]}$, respectively. Upon reflection, we see that
\begin{align}
    B_2^{[1]} &\xrightarrow{R_1} B_2^{[1]}, \nonumber \\
    B_2^{[2]} &\xrightarrow{R_1} -B_2^{[2]}.
\end{align}
This implies that any number of branes coupled to $B_2^{[2]}$ can terminate on the reflection 6-brane. These are M2-branes compactified on $S^1_{(2)}$, i.e., effective strings obtained from wrapped D2-branes.\footnote{Here we take $S^1_{(1)}$ to be the M-theory / IIA circle.}

\paragraph{3-Form Potentials from $C_3$} There is a single candidate 3-form potential in the 9D effective theory that descends from not compactifying $C_3$ on any of the internal directions. We already argued that
\begin{equation}
    C_3 \xrightarrow{R_1} -C_3,
\end{equation}
which implies that any number of M2-branes can terminate on the reflection 6-brane.

\paragraph{Magnetic Dual Branes}

In addition to the branes found above, the magnetic dual branes can also terminate on the reflection 6-brane. These branes couple to $p$-form potentials descending from the M-theory 6-form $\widetilde{C}_6$. In other words, the branes are M5-branes wrapping some number of internal direction. The only difference is that since $\widetilde{C}_6$ is a real 6-form, the M5-brane has to wrap the cycle acted on by the reflection, i.e., $S^1_{(1)}$, in order to pick up a minus sign.

For the case of the reflection 6-brane, the relevant magnetic dual branes couple to $\widetilde{B}_5$, which is the dual of $B_2$, and $\widetilde{C}_4$, which is the dual of $C_3$. From this, we see that any number of 4-branes and 3-branes (coming from M5-branes wrapping either a single or both internal cycles respectively) can terminate on the reflection 6-brane.

\paragraph{Worldvolume Degrees of Freedom} We can use lasso arguments to partially determine the worldvolume degrees of freedom of the reflection 6-brane. To be concrete, take the case of 4-branes terminating on the reflection 6-brane. The worldvolume coupling of the 4-brane is:
\begin{equation}
    \int_{\Sigma_5} B_5.
\end{equation}
This is not gauge invariant under the gauge transformation $B_5 \rightarrow B_5 + d\lambda_4$ when the 4-brane worldvolume has a boundary. The problem comes from the following boundary term
\begin{equation}
    \int_{\partial \Sigma_5} \lambda_4,
\end{equation}
and is important in this case because the boundary of the 4-brane is contained in the reflection 6-brane. We can cancel this boundary term via another coupling term, this time in the reflection 6-brane, that also transforms as $B_5 \rightarrow B_5 + d\lambda_4$:
\begin{equation}
    -\int_{6\textrm{-brane}} \lambda_4 \wedge (\partial \Sigma_5)_{\textrm{PD}},
\end{equation}
where the subscript denotes the Poincar\'e dual. This can come from a term such as
\begin{equation}
    \int_{6\textrm{-brane}} B_5 \wedge f_2,
\end{equation}
where $f_2$ is the field strength for a 1-form gauge field. Upon variation, this term gives
\begin{equation}
    \int_{6\textrm{-brane}} \lambda_4 \wedge df_2.
\end{equation}
From this, we see that if we identify $df_2$ with $(\partial \Sigma_5)_{\textrm{PD}}$, then the configuration will be gauge invariant. Furthermore, this suggests that there is a 1-form gauge field on the worldvolume of the 6-brane. Applying this reasoning to the other relevant $p$-form potentials, we see that there must also be 2-form, 3-form, and 4-form gauge fields on the worldvolume of the reflection 6-brane.

\subsubsection{Generalizations}

We now generalize the previous analysis for the reflection 6-brane to the reflection $(D-3)$-brane, which is the codimension-two defect in the $D$-dimensional effective supergravity theory predicted by the Cobordism Conjecture. Let $d$ denote the number of internal directions. I.e., the 11D spacetime splits as $\mathbb{R}^{D-1,1} \times T^d$. We begin by enumerating all of the $p$-form potentials descending from $C_3$ that undergo a reflection.
Any number of supersymmetric branes coupled to these $p$-form potentials can terminate on the reflection branes.

\paragraph{0-Form Potentials from $C_3$} For $d > 3$, there are $\binom{d-1}{3}$ 0-form potentials that undergo a reflection. These couple to M2-branes wrapping three internal directions, namely pointlike instantons. Strictly speaking, these instantons do not end on the reflection brane, but they become anti-instantons after winding around the reflection brane.

\paragraph{1-Form Potentials from $C_3$} For $d > 2$, there are $\binom{d-1}{2}$ 1-form potentials that undergo a reflection. These couple to M2-branes wrapping two internal directions.

\paragraph{2-Form Potentials from $C_3$} There are $(d-1)$ 2-form potentials that undergo a reflection. These couple to M2-branes wrapping a single internal direction.

\paragraph{3-Form Potentials from $C_3$} There is a single 3-form potential that undergoes a reflection. This is $C_3$ itself, which suggests that any number of M2-branes can terminate on all of the reflection branes.

There are an equal number of magnetic dual branes, descending from $\widetilde{C}_6$, that terminate on the reflection branes. The worldvolume degrees of freedom for each of the reflection branes are determined in exactly the same manner as was done for the reflection 6-brane above. They can also be deduced by compactifiying the reflection 6-brane down to the relevant dimension.

\subsection{Braiding and Binding}

We now consider the interplay of multiple reflection branes. Returning to
our geometric perspective given by M-theory compactified on a square $T^{d}$:%
\begin{equation}
T^{d}=\underset{d\text{ times}}{\underbrace{S^{1}\times...\times S^{1}}},
\end{equation}
we label the reflection brane associated with the $i^{\text{th}}$ factor as
$R_{i} \in G_{U}^{(R)}$, and its Pin$^{+}$-lift by $\widetilde{R}_{i}$.

To begin, let us consider what happens when we have a pair of such branes. If
it is the same sort of brane, then the fact that we have a Pin$^{+}$-lift
means that the corresponding element of $\widetilde{G_{U}}^{+}$ squares to
$1$, namely $\widetilde{R}_{i}^{2}=1$. Physically, we take this to mean that a pair of
such branes annihilate to pure radiation.

Next, suppose we have a pair of such branes $\widetilde{R}_{i}$ and $\widetilde{R}_{j}$ with $i\neq
j$. Since we have singled out two distinguished directions, we can focus our
attention on this $T^{2}$. The general $\widetilde{G_{U}}^{+}$ transformations
descend to the Pin$^{+}$-cover, $\GL^{+}(2,\mathbb{Z})$, a situation that was
analyzed in detail in \cite{Dierigl:2022reg}. In this case, the group theory relations tell us
that:%
\begin{equation}\label{eq:noncomm}
\widetilde{R}_{i}\widetilde{R}_{j}=(-1)^{F}\widetilde{R}_{j}\widetilde{R}_{i}\text{,}%
\end{equation}
where $(-1)^{F}$ is spacetime fermion parity. One way to establish this is to
return to the case of type IIB\ reflection 7-branes. In that setting, one has
an F-theory torus $T^{2}=S^{1}\times S^{1}$ and the two reflections amount to
worldsheet $\mathbb{Z}_{2}$ actions given by left-moving fermion parity $(-1)^{F_{L}}$ and worldsheet orientation reversal $\Omega$. 
Since worldsheet orientation reversal sends left-movers to right-movers, one has $(-1)^{F_{L}%
}\Omega=\Omega(-1)^{F_{R}}=(-1)^{F}\Omega(-1)^{F_{L}}$, where we used the fact
that $(-1)^{F}=(-1)^{F_{L}}(-1)^{F_{R}}$ commutes with $\Omega$ (see also \cite{Tachikawa:2018njr}).

What is the geometry of this $(-1)^{F}$ factor? As explained in
\cite{Dierigl:2022reg}, this amounts to a non-compact elliptically-fibered
$1/2$ K3 surface, i.e., a $dP_{9}$ geometry. Viewed as an elliptic fibration
over a compact $\mathbb{P}^{1}$, the Weierstrass model for $dP_{9}$ is:
\begin{equation}
y^{2}=x^{3}+f_{4}x+g_{6},
\end{equation}
with $f_{4}$ and $g_{6}$ degree $4$ and $6$ polynomials of the homogeneous
coordinates $[z_{1},z_{2}]$ of the base $\mathbb{P}^{1}$.

Observe also that as elements of $\GL(2,\mathbb{Z})$, the combined product of
$R_{i}R_{j}=$ diag$(-1,-1).$ As such, we see that the resulting geometry
produced by a pair of coincident reflection branes is just the orbifold:
\begin{equation}
T^{d-2}\times\left(\mathbb{C}\times T^{2}\right)  / \mathbb{Z}_{2}\text{,}\label{eq:pairR7}
\end{equation}
where the $\mathbb{Z}_{2}$ acts on $\mathbb{C}\times T^{2}$ factors with local coordinates $(z,w)$ as $(z,w)\mapsto(-z,-w)$.
Now, in contrast to the case of F-theory models, the $T^{2}$ in this case is
of finite size. This means there is no \textquotedblleft further
enhancement\textquotedblright\ in the singularity, and we instead have four $A_{1}$ singularities
with local presentation $\mathbb{C}^{2}/\mathbb{Z}_{2}$, i.e., the brane supports an $\mathfrak{su}(2)^{4}$ gauge symmetry, one
gauge algebra for each factor. Observe that this object is supersymmetric
because the orbifold group action preserves the holomorphic $2$-form $dz\wedge
dw$. We take this to mean that there is an attractive potential between a pair
of non-commuting\footnote{See line \eqref{eq:noncomm}.} reflection branes which has this supersymmetric bound state as its end product (accompanied by radiation).

While we have phrased our discussion in terms of M-theory backgrounds, it is
also natural to consider the corresponding F-theory models associated with
these reflection branes. This corresponds to shrinking the $T^{2}$ factor in
line (\ref{eq:pairR7}) to zero size, in which case we reach a non-compact
elliptically-fibered K3 with a singular $I_{0}^{\ast}$ fiber, i.e., we get an
8D $\mathfrak{so}(8)$ gauge theory which is wrapped on a further $S^{1}\times T^{d-2}$,
where the $S^{1}$ factor decompactifies under M- / F-theory duality. This is
essentially the same background as that studied in \cite{Dierigl:2022reg}.

\subsubsection{\texorpdfstring{$M>2$}{M>2} Reflection Branes:\ Supersymmetric Case}

Let us now turn to supersymmetric bound states with more than two
reflection branes. An even number of distinct reflections will produce such
examples. To illustrate, let us consider the case of four reflection branes,
grouped according to their action on the four-torus $T_{(12)}^{2}\times
T_{(34)}^{2}$, where the first factor is associated with monodromy generated
by the pair $R_{1}R_{2}$ (as indicated by the subscript) and similar
conventions for the second factor. Restricting to $\GL(4,\mathbb{Z})$, the
monodromy in this case involves the reflection diag$(-1,-1,-1,-1)$ on a
$T^{4}$ factor.

We now give a geometric characterization of the M-theory background which
realizes this configuration. Observe that if we had not included the
additional $R_{3}R_{4}$ branes, the resulting geometry would be captured by
the quotient:
\begin{equation}
T^{d-4}\times\left(\mathbb{C}\times T_{(12)}^{2}\right)  / \mathbb{Z}_{2}^{(12)}\times T_{(34)}^{2}\text{.}
\end{equation}
Including this extra set of reflection branes amounts to introducing a quotient by another $\mathbb{Z}_{2}$.\footnote{One way to see the presence of an additional quotient is to consider a background where we have separated the two codimension-two defects $R_1 R_2$ and $R_3 R_4$ in the spacetime directions. These are specified by locally independent $\mathbb{Z}_2$ actions.} In other words, the whole singular geometry is of the form:%
\begin{equation}
T^{d-4}\times\left(\mathbb{C}\times T_{(12)}^{2}\times T_{(34)}^{2}\right)  / \mathbb{Z}_{2}^{(12)}\times \mathbb{Z}_{2}^{(34)}\text{,}%
\end{equation}
where the two $\mathbb{Z}_{2}$'s act on the local holomorphic coordinates as:
\begin{align}
\mathbb{Z}_{2}^{(12)} &  :(z,w_{(12)},w_{(34)})\mapsto(-z,-w_{(12)},w_{(34)})\\
\mathbb{Z}_{2}^{(34)} &  :(z,w_{(12)},w_{(34)})\mapsto(-z,w_{(12)},-w_{(34)})\text{.}
\end{align}
This results in sixteen fixed points, all of the same type, locally being
given by $\mathbb{C}^{3}/\mathbb{Z}_{2}\times\mathbb{Z}_{2}$, i.e., we locally get the 5D $T_{2}$ theory, a hypermultiplet in the
trifundamental of the flavor symmetry $\mathfrak{su}(2)^{3}$ (see \cite{Benini:2009gi}). The
further compactification (from working on a compact $T_{(12)}^{2}\times
T_{(34)}^{2}$) means that these flavor symmetries are all gauged, and we are
considering the further dimensional reduction of this theory on a $T^{d-4}$.
This sort of orbifold geometry, including the global form of the gauge group
(including additional Abelian factors)\ was treated in \cite{Cvetic:2023pgm} (see also \cite{Cvetic:2022imb}).

Consider next the F-theory background obtained by shrinking one of these
$T^{2}$ factors, namely we treat $T_{(12)}^{2}$ as the F-theory elliptic
fiber. Observe that in this case, the local collision of singularities
involves $I_{0}^{\ast}$ collisions at the four orbifold fixed points in the
$\mathbb{C}\times T_{(34)}^{2}$ directions, namely we get an $\mathfrak{so}%
(8)^{4}$ global symmetry with pairwise collisions resulting in $(D_{4},D_{4})$
conformal matter, as in references \cite{Heckman:2013pva, DelZotto:2014hpa,
Heckman:2014qba, Heckman:2015bfa} (see \cite{Heckman:2018jxk, Argyres:2022mnu}
for reviews).

Similar considerations hold for additional supersymmetric combinations of
reflection branes. For an even number of reflections on an internal torus
$T^{2k}$, the resulting monodromy for the codimension-two defect in
$\GL(2k,\mathbb{Z})$ is generated by the element:%
\begin{equation}
\text{diag}(\underset{2k\text{ times}}{\underbrace{-1,...,-1}})\in
\text{GL}(2k,\mathbb{Z}) \,.
\end{equation}
In this more general case, the supersymmetric M-theory background is:%
\begin{equation}
T^{d-2k}\times\left(\mathbb{C}\times T^{2k}\right)  /\Gamma\text{,}
\end{equation}
where the group $\Gamma$ is:%
\begin{equation}
\Gamma=\underset{k\text{ times}}{\underbrace{
\mathbb{Z}_{2}\times...\mathbb{Z}_{2}}},
\end{equation}
where each $
\mathbb{Z}_{2}$ acts via a sign flip on the $z$ coordinate of $\mathbb{C}$ and one of the holomorphic $T^{2}$ factors of:
\begin{equation}
T^{2k}=\underset{k\text{ times}}{\underbrace{T^{2}\times...T^{2}}}.
\end{equation}

As an example along these lines, consider $k=3$. This results in an M-theory
background with local $\mathbb{C}^{4}/%
\mathbb{Z}
_{2}\times%
\mathbb{Z}
_{2}\times%
\mathbb{Z}
_{2}$ singularities. This local singularity structure results in a
3D\ $\mathcal{N}=2$ theory with an $\mathfrak{su}(2)^{6}$ flavor symmetry
(fixed divisors coming from pairs of local equations of the form $z_{i}%
=z_{j}=0$ for $i\neq j$) and four matter fields in tri-fundamental
representations (fixed curves coming from triples of local equations of the
form $z_{i}=z_{j}=z_{k}=0$ for $i,j,k$ distinct) and a localized interaction
term at $z_{i}=0$ for all $i$ with superpotential of the schematic form
$W=X_{(123)}X_{(124)}X_{(134)}X_{(234)}$ (which is classically marginal in 3D). We summarize the matter
content for this local model below, where subscripts denote directions in which the associated singular loci are localized:%
\begin{equation}%
\begin{tabular}
[c]{|c|c|c|c|c|c|c|}\hline
& $\mathfrak{su}(2)_{(12)}$ & $\mathfrak{su}(2)_{(13)}$ & $\mathfrak{su}%
(2)_{(14)}$ & $\mathfrak{su}(2)_{(23)}$ & $\mathfrak{su}(2)_{(24)}$ &
$\mathfrak{su}(2)_{(34)}$\\\hline
$X_{(123)}$ & $\mathbf{2}$ & $\mathbf{2}$ & $\mathbf{\cdot}$ & $\mathbf{2}$ &
$\mathbf{\cdot}$ & $\mathbf{\cdot}$\\\hline
$X_{(124)}$ & $\mathbf{2}$ & $\cdot$ & $\mathbf{2}$ & $\cdot$ & $\mathbf{2}$ &
$\cdot$\\\hline
$X_{(134)}$ & $\cdot$ & $\mathbf{2}$ & $\mathbf{2}$ & $\cdot$ & $\cdot$ &
$\mathbf{2}$\\\hline
$X_{(234)}$ & $\mathbf{\cdot}$ & $\cdot$ & $\cdot$ & $\mathbf{2}$ &
$\mathbf{2}$ & $\mathbf{2}$\\\hline
\end{tabular}
.
\end{equation}
\newline On compact tori we simply get additional copies of this same system
with gauged combinations of the original flavor symmetries.

Likewise, the F-theory lift of the local model $(\mathbb{C}^{3}\times T^{2})/\mathbb{Z}_{2} \times\mathbb{Z}_{2} \times\mathbb{Z}_{2}$ involves a triple intersection of $\mathfrak{so}(8)$ flavor symmetries,
with local Weierstrass model of the form:%
\begin{equation}
y^{2}=x^{3}+\alpha x(u_{1}u_{2}u_{3})^{2}+\beta(u_{1}u_{2}u_{3})^{3},
\end{equation}
i.e., it is an example of an $\mathfrak{so}(8)^{3}$ conformal Yukawa in the sense of reference
\cite{Apruzzi:2018oge}.\footnote{Note in particular that these further
collisions of singularities do not result in matter in even bigger
representations of the flavor symmetry, but rather additional interactions (see e.g, \cite{Apruzzi:2016iac}). Thus, this is in accord with the conjectured absence of an isolated 5-plet of $\mathfrak{su}(2)$ discussed in \cite{Baumgart:2024ezp}.}

Compactifying all the way to 3D is the lowest we can go while still retaining
a sensible notion of a codimension-two object.

\subsubsection{\texorpdfstring{$M>2$}{M>2} Reflection Branes:\ Non-Supersymmetric Case}

Consider next the case of configurations involving an odd number of
reflection branes, say $2k+1$. In this case, we do not preserve any
supersymmetries. Topologically, the geometry in this case involves a
combination of the supersymmetry preserving quotient by the group $\Gamma=\mathbb{Z}_{2}^{k}$, as well as a Klein bottle $\mathrm{KB}_{\infty}$, viewed as a circle bundle over the spacetime $S_{\infty}^{1} = \partial \mathbb{C}$. In this case, the full geometry takes the form:%
\begin{equation}
T^{d-(2k+1)}\times\text{Cone}(\mathrm{KB}_{\infty} \times T^{2k})/\Gamma\text{.}
\end{equation}
One way to understand these examples is to start with one of our supersymmetric
backgrounds and simply add an additional non-supersymmetric reflection brane.
We expect that this engineers an interacting quantum field theory.

\section{Conclusions} \label{sec:CONC}

Dualities provide important constraints on the non-perturbative structure of
quantum theories. In this work we have determined the Spin- and Pin$^+$-lifts of
the U-dualities of maximally supersymmetric non-chiral supergravity theories.
Using this, we applied the Swampland Cobordism Conjecture to predict the
existence of codimension-two reflection branes. These branes are
lower-dimensional analogs of the reflection 7-branes found in type II\ string
theory. Indeed, the reflection branes found in this paper arise from wrapping
R7-branes on higher-dimensional cycles of the internal torus of an M-theory
compactification. We have also argued that these branes support non-trivial
degrees of freedom since BPS\ branes can terminate on them, and moreover, have
also established some basic features such as braiding and bound state
formation. In the remainder of this section we discuss some potential avenues
of future investigation.

One of the motivations for the present work was to determine the spectrum of
objects predicted by the Swampland Cobordism Conjecture. Now that we have
determined the full U-duality group, it is natural to return to the question
of the corresponding bordism groups $\Omega_{k}^{\Spin\text{-}\widetilde{G_{U}%
}}(\pt)$ and $\Omega_{k}^{\Spin\text{-}\widetilde{G_{U}}^{+}}(\pt)$ and
extract predictions for non-perturbative objects of these gravitational
theories. Especially in the case of Spin-$\widetilde{G_{U}}^{+}$ bordisms, we
expect that some of these defects will be non-supersymmetric.

Another natural extension would be to consider even more general tangential structures on our $D$-dimensional effective field theories. While we still required a $\Spin\text{-} \Gamma$ structure for our spacetime, it would be interesting to also consider further refinements, as motivated by M-theory, such as suitable twistings of Pin, String and other related structures.

It would also be interesting to determine the corresponding Spin- and
Pin$^{+}$-lifts for systems with reduced supersymmetry. For example Calabi-Yau compactifications often have non-trivial duality groups inherited from the automorphisms of the Calabi-Yau manifold \cite{Delgado:2024skw}. One could thus carry out an analysis of (possiby non-supersymmetric) objects predicted by the Cobordism Conjecture.

We have primarily used group-theoretic and topological considerations to argue
for the existence of these reflection branes and to determine their properties. It would be
interesting to work out the corresponding supergravity solutions which include
these defects. Among other things, this would allow us to extract their tension.

Wrapping branes on \textquotedblleft cycles at infinity\textquotedblright\ has
been a fruitful way to engineer a wide class of topological symmetry operators
in stringy QFTs as well as holographic systems.\footnote{See e.g., \cite{Apruzzi:2022rei, GarciaEtxebarria:2022vzq, Heckman:2022muc, Heckman:2022xgu, Heckman:2024oot, Heckman:2025lmw}.} It would be interesting to study how these branes realize discrete symmetry operators, perhaps along the lines of \cite{Dierigl:2023jdp}.

\section*{Acknowledgements}

We thank N. Braeger and M. Montero for collaboration at an
early stage of this work. We also thank D.S. Berman, G. Bossard, N. Braeger,  C.M. Hull, J. McNamara, M. Montero, S. Raman, Y. Tachikawa, and E. Torres for helpful discussions. VC\ and JJH\ thank the 2025 Simons Summer workshop for hospitality during the completion of this work.
JJH thanks Mountain Dew for continuing to provide an excellent selection of thirst quenching
products with bold citrus flavor, including Mountain Dew Original; Mountain Dew Code Red; Mountain Dew Voltage; Mountain Dew Livewire; and Mountain Dew Baja Blast \cite{DEW}.
The work of VC is supported by an NSF Graduate Research Fellowship.
The work of JJH is supported by DOE (HEP) Award DE-SC0013528 as well as by BSF
grant 2022100. The work of JJH is also supported in part by a University
Research Foundation grant at the University of Pennsylvania.

\appendix

\section{Reflections on the States of M-Theory} \label{app:MASSIVE}

We now explain in greater detail why the reflections of M-theory require us to
enlarge the U-duality groups. To this end, let us begin in type IIA\ string
theory compactified on a $T^{d-1}$. In this case, the T-duality symmetry group is
$\mathrm{Spin}(d-1,d-1,\mathbb{Z})$, with possible discrete quotients. The T-duality group contains $\SL(d-1,\mathbb{Z})$, the group of large diffeomorphisms on the $T^{d-1}$. From the perspective
of the low energy supergravity theory, there is little difference between
IIA and IIB, and so one can also entertain the $\mathrm{det}=-1$ elements which amount to inverting the length of a circle, namely $L \mapsto 1/L$.
Additionally, it is worth noting that the
RR\ states transform in spinor representations of the T-duality group, and so 
we have written $\textrm{Spin(}d-1,d-1,\mathbb{Z})$. We emphasize that this is still a statement purely connected with the
bosonic sector of the theory, and is not directly associated with the
fermionic degrees of freedom (which require a Spin-lift of the U-duality group).

Now, in addition to these T-duality symmetries, we also have $(-1)^{F_{L}}$,
namely left-moving fermion parity of the IIA\ theory. This is an additional $\mathbb{Z}_{2}$ symmetry and acts by sending RR\ fields to minus themselves. In terms of
M-theory on $S^{1}\times T^{d-1}$, with reduction on $S^{1}$ taking us to
type\ IIA, this $\mathbb{Z}_{2}$ corresponds to the reflection $\theta\mapsto-\theta$ on the local coordinate.

It is instructive to see how this reflection acts on the spectrum of massless
and massive states of the theory. To illustrate, we focus on a toroidal compactification
in which we impose periodic boundary conditions for bosonic fields. As a representative example, we consider the spectrum of two-index anti-symmetric tensor fields, as obtained from dimensional reduction of the pseudo 3-form potential $C_{3}$ and its magnetic dual real 6-form potential $\widetilde{C}_6$.\footnote{Recall that a pseudo-form transforms
with an extra minus sign under a reflection.} We also retain the
explicit internal dependence on $\theta$ by splitting up our potentials into
reflection even Fourier modes and reflection odd Fourier modes, namely
$\cos n \theta$ and $\sin n \theta$, which we denote by $m\in\{0,1\}$, with $m = 0$ for the parity 
even modes (cosines) and $m = 1$ for the parity odd modes (sines). 
We introduce the notation $C_{3}(m)$ and $\widetilde{C}_6(m)$ to capture these modes.
Note that the massless sector is in $m=0$, but that $m=0$ also includes
massive excitations. For $m=1$ all modes are massive.

Let us explore the consequences of this in an explicit example, with M-theory
compactified on a $T^{d}$. We mostly keep the discussion general for arbitrary
$d$, but specialize to $d=4$ when convenient to illustrate the main ideas
since similar considerations hold for more general cases.

Consider first the reduction of $C_{3}(m)$ on our $T^{d}$. Observe that in the
$D$-dimensional spacetime, we get two-index anti-symmetric tensor fields by
keeping one leg internal. This results in $d$ such fields which we write as
$C_{\mu\nu i}(m)$. Reflections on the M-theory circle further partition up
these degrees of freedom. Since we are dealing with a pseudo 3-form, we
have:%
\begin{align}
R_{1}  & :C_{\mu\nu1}(m) \mapsto (-1)^{m} \, C_{\mu\nu1}(m)\\
R_{1}  & :C_{\mu\nu i}(m) \mapsto (-1)^{m+1} \, C_{\mu\nu i}(m)\text{ \ \ for
\ \ }i\neq1\text{.}%
\end{align}

Consider next the reduction of $\widetilde{C}_6(m)$ on our $T^{d}$. To get a
2-index anti-symmetric tensor field, four indices must be kept internal, so
this only contributes when $d>3$. When this holds, we have $\binom{d}{4}%
=\frac{d(d-1)(d-2)(d-3)}{12}$ such fields. Of these, there is a further
refinement depending on their sign under an internal reflection. In
particular, we have, for $i,j,k,l$ all distinct:%
\begin{align}
R_{1}  & :\widetilde{C}_{\mu\nu1ijk}(m) \mapsto (-1)^{m+1} \, \widetilde{C}%
_{\mu\nu1ijk}(m)\text{ \ \ for \ \ }i,j,k\neq1\\
R_{1}  & :\widetilde{C}_{\mu\nu ijkl}(m) \mapsto (-1)^{m} \, \widetilde{C}%
_{\mu\nu ijkl}(m)\text{\ \ \ \ \ \, for \ \ }i,j,k,l\neq1\text{.}%
\end{align}
Totaling everything up, we now see that our states organize according to
their sign under reflections. In particular, the same sign under reflection is
fixed for:
\begin{align}
(-1)^{m}\text{ parity}  & \text{: \ \ }C_{\mu\nu1}(m)\text{ and }%
\widetilde{C}_{\mu\nu ijkl}(m)\text{ \ \ for \ \ }i,j,k,l\neq1\\
(-1)^{m+1}\text{ parity}  & \text{: \ \ }C_{\mu\nu i}(m)\text{ and
}\widetilde{C}_{\mu\nu1ijk}(m)\text{ \ \ for \ \ }i,j,k\neq1.
\end{align}

Let us now specialize further to $d=4$. Fixing the overall mass, we see that for each mass we have the following number of tensor fields:%
\begin{align}
(-1)^{m}\text{ parity}  & \text{: \ \ \ \ }C_{\mu\nu1}(m)\Rightarrow 1\text{
\, tensor field}\\
(-1)^{m+1}\text{ parity}  & \text{: \ \ \ \ }C_{\mu\nu i}(m)\text{ and
}\widetilde{C}_{\mu\nu1234}(m)\Rightarrow 4\text{ \, tensor fields,}%
\end{align}
which have opposite signs under an internal reflection. Of course, the
massless sector simply reproduces the expected spectrum of IIA, where
$C_{\mu\nu1}$ descends to the familiar NSNS 2-form potential (inert under
$(-1)^{F_{L}}$), while the remaining massless potentials $C_{\mu\nu i}$ and
$\widetilde{C}_{\mu\nu1234}$ are the descendants of RR\ potentials and
transform in the four-dimensional spinor representation of Spin$(3,3,\mathbb{Z})$.

By inspection, we see that as expected, the massless modes respect the
breaking pattern in reducing from M-theory to IIA, namely $\SL(5,\mathbb{Z})\supset \Spin(3,3,\mathbb{Z})\simeq \SL(4,\mathbb{Z})$ where the vector representation decomposes as $\mathbf{5}\rightarrow
\mathbf{4}\oplus\mathbf{1}$, as expected. Observe also that on these massless
states, $(-1)^{F_{L}}$ embeds as diag$(+1,-1,-1,-1,-1)$, namely it appears as
an element of $\SL(5,\mathbb{Z})$. On the other hand, for the first massive excitations which are odd under
reflections, we again have a decomposition into representations of
$\mathbf{5}\rightarrow\mathbf{4}\oplus\mathbf{1}$, but $(-1)^{F_{L}}$ embeds
as diag$(-1,+1,+1,+1,+1)$, namely an element of $\GL(5,\mathbb{Z})$ but not $\SL(5,\mathbb{Z})$.

Similar considerations hold for other toroidal compactifications, and is in
line with the general expectation that we cannot generate $(-1)^{F_{L}}$ from
large diffeomorphisms of M-theory in tandem with T-dualities.\ As such, we
must extend the U-duality group to include these symmetries.

\section{Split Real Form versus Compact Real Form} \label{app:Split_vs_Compact}

U-duality groups naturally arise from toroidal compactifications of M-theory. At the level of supergravity, this results in $G_{U}(\mathbb{R})$, the split real form of a complex Lie group. As proposed in \cite{Hull:1994ys}, the refinement to a quantized spectrum of objects results in $G_{U}(\mathbb{Z}) \equiv G_{U}$, with a natural embedding $i: G_{U}(\mathbb{Z}) \rightarrow G_{U}(\mathbb{R})$. A simple example is the inclusion $\SL(2,\mathbb{Z}) \rightarrow \SL(2,\mathbb{R})$. At the other extreme, we have the U-duality group in three dimensions $G_{U}^{\mathrm{3D}} = E_{8(8)}(\mathbb{Z}) \rightarrow E_{8(8)}(\mathbb{R})$.

We are interested in possible Spin- and Pin$^{+}$-lifts of our U-dualities, as required by the fermionic degrees of freedom of our theory. We use the same strategy deployed in \cite{Pantev:2016nze} and \cite{Tachikawa:2018njr}, namely we first study extensions of $G_{U}(\mathbb{R})$ and then show that this induces an extension of $G_{U} = G_{U}(\mathbb{Z})$.

The first important comment is that as opposed to the compact real forms, the split real form of Lie groups do have non-trivial $\mathbb{Z}_2$ extensions \cite{Knapp:2002}.\footnote{Indeed, the only $\mathbb{Z}_2$ extension of $E^{\mathrm{cpct}}_8$ is the trivial one to $E_{8}^{\mathrm{cpct}} \times \mathbb{Z}_2$. For the split real form $E_{8(8)}(\mathbb{R})$, a non-trivial extension is possible.}

For ease of exposition, we focus on the case of 3D supergravity with U-duality group $G_{U}(\mathbb{R}) = E_{8(8)}(\mathbb{R})$. Similar considerations hold for the $D > 3$ U-duality groups, and can also be deduced by taking suitable decompactification limits.

The existence of a non-trivial central extension, such as a Spin-lift, is governed by the topology of $G_U(\mathbb{R})$. More precisely, central extensions of the U-duality group $G_U(\mathbb{R})$, which is a connected group, by a discrete Abelian group $A$ (such as $\mathbb{Z}_2$) are classified by the second group cohomology $H^2(G_U(\mathbb{R}); A)$, which is in turn encoded by the fundamental group:
\begin{equation}
H^2(G_U(\mathbb{R}); A) \cong \operatorname{Hom}(\pi_1(G_U(\mathbb{R})), A).
\end{equation}
Thus, a non-trivial fundamental group can give rise to non-trivial central extensions. Since we are interested in the double cover / Spin-lift of $G_U(\mathbb{R})$, this simplifies to
\begin{equation}
H^2(G_U(\mathbb{R}); \mathbb{Z}_2) \cong \operatorname{Hom}(\pi_1(G_U(\mathbb{R})), \mathbb{Z}_2).
\end{equation}

For the split real form of an exceptional Lie group, such as $E_{8(8)}(\mathbb{R})$, we have that $\pi_1(G_U(\mathbb{R})) = \pi_1(K_U)$ (see Appendix \ref{app:Spin_Pin_Lifts} for more details). In this case, $K_U^{\textrm{3D}} = \Spin(16)/\mathbb{Z}_2$:
\begin{equation}
\pi_1(K_U^{\textrm{3D}}) = \pi_1(G_U^{\textrm{3D}}(\mathbb{R})) = \mathbb{Z}_2.
\end{equation}
As a result, there exists a non-trivial double cover $\Spin(16) \to \Spin(16)/\mathbb{Z}_2$, and one is led to consider a non-trivial $\mathbb{Z}_2$ central extension $\widetilde{E}_{8(8)}(\mathbb{R})$ in theories that include fermions.

By contrast, the compact real form of $E_8$ is simply connected:
\begin{equation}
\pi_1(E_8^{\textrm{cpct}}) = 0.
\end{equation}
This implies that all central extensions of $E_8^{\textrm{cpct}}$ are trivial. That is, any central extension of the form
\begin{equation}
1 \to \mathbb{Z}_2 \to \widetilde{E}_8 \to E_8^\textrm{cpct} \to 1
\end{equation}
splits as a direct product:
\begin{equation}
\widetilde{E}_8 \cong E_8^\textrm{cpct} \times \mathbb{Z}_2.
\end{equation}

In summary, the split real form $E_{8(8)}$ admits a non-trivial Spin-lift due to the non-trivial topology of the maximal compact subgroups. On the other hand, the compact real form $E_8^\textrm{cpct}$ is simply connected and admits no non-trivial central extensions.

This analysis extends to the U-duality groups in 4D, given by $G^{\mathrm{4D}}_{U}(\mathbb{R}) = E_{7(7)}(\mathbb{R})$, and 5D, given by $G^{\mathrm{5D}}_{U}(\mathbb{R}) = E_{6(6)}(\mathbb{R})$. The maximal compact subgroups can be found in Table \ref{tab:Spin_U-duality_Summary} and are given by
\begin{equation}
    K^{\mathrm{4D}}_{U} = \USp(8)/\mathbb{Z}_2 \qquad \textrm{and} \qquad K^{\mathrm{5D}}_{U} = \SU(8)/\mathbb{Z}_2.
\end{equation}
From this we see that $\pi_1(G^{\mathrm{4D}}_{U}) = \pi_1(K^{\mathrm{4D}}_{U}) = \mathbb{Z}_2$ and $\pi_1(G^{\mathrm{5D}}_{U}) = \pi_1(K^{\mathrm{5D}}_{U}) = \mathbb{Z}_2$, which indicates that $E_{6(6)}$ and $E_{7(7)}$ admit non-trivial Spin-lifts. In contrast, the compact real forms $E_7^\textrm{cpct}$ and $E_6^\textrm{cpct}$ are simply connected.\footnote{We use the simply connected form of the exceptional Lie group to construct the real forms.} Thus, the fundamental groups are trivial in each case, and any candidate central extension splits as a direct product.

\section{Explicit Spin- / \texorpdfstring{Pin$^+$}{Pin+}-Lifts} \label{app:Spin_Pin_Lifts}

In this Appendix we explicitly construct the Spin- and Pin$^+$-lifts of the bosonic U-duality groups. The Spin-lifts are given by a non-trivial $\mathbb{Z}_2$ extension of the original U-duality groups, while the Pin$^+$-lift is given by including an additional reflection generator in the disconnected component of the U-duality group. The reflection generator acts on the other generators of the U-duality group via conjugation. We also comment on the decompactification limits of the extended U-duality groups.

\subsection{Spin-Lifts} \label{App:Explicit_Spin_Lifts}

The discrete bosonic U-duality group $G_U(\mathbb{Z}) \equiv G_U$ of the $D$-dimensional effective theory arising in toroidal compactifications of maximal supergravity does not act linearly on fermionic fields. To define a consistent duality action on spinors, one must replace $G_U(\mathbb{Z})$ with its double cover. This is given by a non-trivial central extension of $G_U(\mathbb{Z})$ by $\mathbb{Z}_2$ known as the Spin-lift.

The Spin-lift $\widetilde{G_U}(\mathbb{Z})$ acts on fermionic states via linear representations. When $D=9$, this construction recovers the metaplectic group $\Mp(2,\mathbb{Z})$ as the Spin-lift of $\SL(2,\mathbb{Z})$. At higher rank, it defines unique $\mathbb{Z}_2$ central extensions of groups such as $E_{6(6)}(\mathbb{Z})$, $E_{7(7)}(\mathbb{Z})$, and $E_{8(8)}(\mathbb{Z})$.

Following the construction of Pantev and Sharpe \cite{Pantev:2016nze}, the double cover $\widetilde{G_U}(\mathbb{Z})$ is defined as the pullback of the universal cover $\widetilde{G_U}(\mathbb{R}) \to G_U(\mathbb{R})$ along the inclusion $G_U(\mathbb{Z}) \hookrightarrow G_U(\mathbb{R})$. Explicitly, this gives:
\begin{equation}
\widetilde{G_U}(\mathbb{Z}) := \left\{ (a, g) \in \widetilde{G_U}(\mathbb{R}) \times G_U(\mathbb{Z}) \;\middle|\; p(a) = i(g) \right\},
\end{equation}
where $p$ is the covering map and $i$ is the inclusion. Since $\pi_1(G_U(\mathbb{R})) \supset \mathbb{Z}_2$ in almost all cases, we see that
\begin{equation}
1 \rightarrow \mathbb{Z}_2 \rightarrow \widetilde{G_U}(\mathbb{R}) \rightarrow G_U(\mathbb{R}) \rightarrow 1. \label{eq:Spin_Lift_Real}
\end{equation}
This can be seen by computing $\pi_1(K_U)$, where $K_U$ is the maximal compact subgroup of $G_U(\mathbb{R})$. In particular, in all cases $G_U(\mathbb{R})$ is a connected, real group and has finite center. Furthermore, $G_U(\mathbb{R})/K_U$ is contractible in all cases. Thus, the inclusion $K_U \hookrightarrow G_U$ induces an isomorphism
\begin{equation}
    \pi_1(K_U) = \pi_1(G_U(\mathbb{R))}.
\end{equation}
From this, we can define the desired Spin-lift / double cover of $G_U(\mathbb{Z})$ via pullback from line \eqref{eq:Spin_Lift_Real}:
\begin{equation}
    1 \rightarrow \mathbb{Z}_2 \rightarrow \widetilde{G_U}(\mathbb{Z}) \rightarrow G_U(\mathbb{Z}) \rightarrow 1. \label{eq:Spin_Lift_Discrete}
\end{equation}
For the 9D U-duality group $G^\mathrm{9D}_U(\mathbb{Z})= \SL(2,\mathbb{Z})$, the story is slightly different: $\pi_1(K_U^{\textrm{9D}}=\SO(2)) = \mathbb{Z}$, so we use not the universal cover, but the unique double cover of $K_U^{\textrm{9D}}$ to obtain the metaplectic double cover $\Mp(2,\mathbb{Z})$.

The Spin-lift of the bosonic U-duality groups in 8D, $G_U^{\mathrm{8D}}(\mathbb{Z})= \SL(3,\mathbb{Z}) \times \SL(2,\mathbb{Z})$, while still given by \eqref{eq:Spin_Lift_Discrete}, has some additional subtleties. $G_U^\mathrm{8D}$ is a product of two groups, both of which have a universal cover. Thus, the correct Spin-lift / double cover is found by extending both groups and then quotienting by a diagonal $\mathbb{Z}_2^{\textrm{diag}}$, as can be verified by comparing the Spin-lifts of the U-duality groups across different dimensions.

See Table \ref{tab:Spin_U-duality_Summary} for a summary of the continuous and discrete bosonic U-duality groups, the corresponding maximal compact subgroups, and the Spin-lifts of each. This universal construction provides a systematic and dimension-independent method for determining the correct U-duality symmetry group acting on all fields, including fermions, in maximal supergravity.

\begin{table}[!ht]
\centering
\scalebox{0.7}{
\begin{tabular}{|c|c|c|c|c|c|}
\hline
\textbf{$D$} & \makecell{\textbf{Classical U-duality}\\\textbf{Group $G_U(\mathbb{R})$}} & \makecell{\textbf{Discrete U-duality}\\\textbf{Group $G_U(\mathbb{Z})$}} & \makecell{\textbf{Maximal Compact}\\\textbf{Subgroup $K_U$}} & \makecell{\textbf{Spin Lift}\\$\widetilde{K_U}$} & \makecell{\textbf{Spin Lift}\\$\widetilde{G_U}(\mathbb{Z)}$} \\
\hline
$9$ & $\SL(2,\mathbb{R})$ & $\SL(2,\mathbb{Z})$ & $\SO(2)$ & $\Spin(2)$ & $\Mp(2,\mathbb{Z})$ \\[0.3cm]
$8$ & $\SL(3,\mathbb{R}) \times \SL(2,\mathbb{R})$ & $\SL(3,\mathbb{Z}) \times \SL(2,\mathbb{Z})$ & $\SO(3) \times \SO(2)$ & $\Spin(3) \times \Spin(2)/\mathbb{Z}_2^\textrm{diag}$ & $\widetilde{\SL}(3,\mathbb{Z}) \times \Mp(2,\mathbb{Z})/\mathbb{Z}_2^\textrm{diag}$ \\[0.3cm]
$7$ & $\SL(5,\mathbb{R})$ & $\SL(5,\mathbb{Z})$ & $\SO(5)$ & $\Spin(5)$ & $\widetilde{\SL}(5,\mathbb{Z})$ \\[0.3cm]
$6$ & $\Spin(5,5,\mathbb{R})$ & $\Spin(5,5,\mathbb{Z})$ & $\Spin(5) \times \Spin(5)/\mathbb{Z}_2^{\textrm{diag}}$ & $\Spin(5) \times \Spin(5)$ & $\widetilde{\Spin}(5,5,\mathbb{Z})$ \\[0.3cm]
$5$ & $E_{6(6)}$ & $E_{6(6)}(\mathbb{Z})$ & $\USp(8)/\mathbb{Z}_2$ & $\USp(8)$ & $\widetilde{E}_{6(6)}(\mathbb{Z})$ \\[0.3cm]
$4$ & $E_{7(7)}$ & $E_{7(7)}(\mathbb{Z})$ & $\SU(8)/\mathbb{Z}_2$ & $\SU(8)$ & $\widetilde{E}_{7(7)}(\mathbb{Z})$ \\[0.5cm]
$3$ & $E_{8(8)}$ & $E_{8(8)}(\mathbb{Z})$ & $\Spin(16)/\mathbb{Z}_2$ & $\Spin(16)$ & $\widetilde{E}_{8(8)}(\mathbb{Z})$ \\[0.3cm]
\hline
\end{tabular}
}
\caption{The classical U-duality groups $G_U(\mathbb{R})$, the discrete U-duality groups $G_U(\mathbb{Z})$, the corresponding maximal compact subgroups $K_U$, and the respective Spin-lifts $\widetilde{K_U}$ and $\widetilde{G_U}$ appearing in $D$-dimensional supergravity theories for $3 \leq D \leq 9$.}
\label{tab:Spin_U-duality_Summary}
\end{table}

\subsection{\texorpdfstring{Pin$^+$}{Pin+}-Lifts} \label{App:Explicit_Pin_Lifts}

In addition to orientation-preserving symmetries, physical duality groups often include elements that reverse orientation, such as spacetime or internal reflections. These generate an extension of the bosonic U-duality group by a discrete reflection symmetry, resulting in the semi-direct product
\begin{equation}
G_U \rtimes \mathbb{Z}_2^R,
\end{equation}
where the $\mathbb{Z}_2$ factor corresponds to a chosen reflection representative. Such reflections lie outside the identity component of $K_{U}^{(R)} \supset G_{U}^{(R)}(\mathbb{R})$, the maximal compact subgroup. More precisely, including such a reflection element extends the maximal compact subgroup $K_U$ as determined by the short exact sequence:
\begin{equation}
    1 \rightarrow K_U \rightarrow K_U \rtimes \mathbb{Z}_2^R \rightarrow \mathbb{Z}_2^R \rightarrow 1.
\end{equation}
Then, via the inclusion $K_U \hookrightarrow G_U(\mathbb{R})$, this also induces an extension of $G_U(\mathbb{R})$ as determined by a similar short exact sequence:
\begin{equation}
    1 \rightarrow G_U(\mathbb{R}) \rightarrow G_U(\mathbb{R}) \rtimes \mathbb{Z}_2^R \rightarrow \mathbb{Z}_2^R \rightarrow 1. \label{eq:Reflections_GU_Cont}
\end{equation}
From this, we again get the desired extension on the discrete group via a pullback:
\begin{equation}
    1 \rightarrow G_U \rightarrow G_U \rtimes \mathbb{Z}_2^R \rightarrow \mathbb{Z}_2^R \rightarrow 1. \label{eq:Reflections_GU_Disc}
\end{equation}

To accommodate fermions in the presence of such orientation-reversing symmetries, one must lift this group to a central extension that incorporates both the Spin and reflection structure, i.e., a Pin$^+$-lift. This is given by combining the extensions \eqref{eq:Spin_Lift_Discrete} and \eqref{eq:Reflections_GU_Disc}:
\begin{equation}
1 \rightarrow \mathbb{Z}_2 \rightarrow \widetilde{G_U}^+ \rightarrow G_U \rtimes \mathbb{Z}_2^R \rightarrow 1.
\end{equation}
This lift is governed by the structure of the $\Pin^+$ group, which is the double cover of the full orthogonal group $\O(n)$, just as $\Spin(n)$ is the double cover of $\SO(n)$. In particular, a reflection element $R \in \O(n) \setminus \SO(n)$, when lifted to an element $\widetilde R$ in the $\Pin^+$ group, satisfies $\widetilde{R}^2 = 1$ on spinors.

In each dimension, we include only a single reflection generator to extend the U-duality group because the relevant outer automorphism group is $\mathbb{Z}_2^R$, corresponding to the disconnected component of the full U-duality group (e.g., $\GL(n,\mathbb{Z})$ versus $\SL(n,\mathbb{Z})$). Although there are many reflection-like elements in the full group, they are all conjugate to each other, so their effect on the U-duality group is captured by a single non-trivial automorphism. Hence, adjoining one reflection generator that implements this outer action suffices to generate the full semi-direct product structure. This also ensures the minimal and correct extension when considering spinor representations.

In practice, this minimal extension is achieved by performing a Pin$^+$-lift of the maximal compact subgroup $K_U$ of $G_U(\mathbb{R})$, which in turn gives the correct extension of $G_U(\mathbb{R})$ via embedding $\widetilde{K_U}^+$, and from this the correct extension on $G_U$ itself via a pullback. All of the relevant extensions are summarized in Table \ref{tab:Pin+_U-duality_Summary}. Note that the case of $D=8$ is more subtle, and is treated with more care in the next section. The case of $D=6$ is also subtle, as the maximal compact subgroup is the product of two groups: $K_U^{5D} = (\Spin(5) \times \Spin(5))/\mathbb{Z}_2$. Since the two $\Spin(5)$ factors embed into $G_U(\mathbb{R}) = \Spin(5,5)$ block diagonally, it is important to only extend one of the $\Spin(5)$ factors to have an overall element with determinant $-1$. It does not matter which factor gets extended, as they can be related via conjugation.

\begin{table}[!ht]
\centering
\scalebox{0.82}{
\begin{tabular}{|c|c|c|}
\hline
\textbf{$D$} & \makecell{\textbf{Maximal Compact}\\\textbf{Subgroup $K_U \subset G_{U}(\mathbb{R})$}} & $G^{(R)}_U = G_U \rtimes \mathbb{Z}_2^R$ \\
\hline
$9$ & $\SO(2)$ & $\SL(2,\mathbb{Z})\rtimes \mathbb{Z}_2^R$ \\[0.2cm]

$8$ & $\SO(3) \times \SO(2)$ & $(\SL(3,\mathbb{Z}) \times \SL(2,\mathbb{Z})) \rtimes \mathbb{Z}_2^R$ \\[0.2cm]

$7$ & $\SO(5)$ & $\SL(5,\mathbb{Z}) \rtimes \mathbb{Z}_2^R$ \\[0.2cm]

$6$ & $\Spin(5) \times \Spin(5)/\mathbb{Z}_2^\textrm{diag}$ & $\Spin(5,5,\mathbb{Z}) \rtimes \mathbb{Z}_2^R$ \\[0.2cm]

$5$ & $\USp(8)/\mathbb{Z}_2$ & $E_{6(6)}(\mathbb{Z}) \rtimes \mathbb{Z}_2^R$ \\[0.2cm]

$4$ & $\SU(8)/\mathbb{Z}_2$ & $E_{7(7)}(\mathbb{Z}) \rtimes \mathbb{Z}_2^R$ \\[0.2cm]

$3$ & $\Spin(16)/\mathbb{Z}_2$ & $E_{8(8)}(\mathbb{Z}) \rtimes \mathbb{Z}_2^R$ \\[0.2cm]
\hline
\end{tabular}
}
\caption{The maximal compact subgroup $K_U$ of the classical U-duality group, and the lift of the bosonic U-duality group $G_U^{(R)}$ for $3 \leq D \leq 9$. Extending $K_U$ to include reflections induces a lift of $G_U$ to $G_U^{(R)}$. The full Pin$^+$-lift of $G_U$ is given by a $\mathbb{Z}_2$ central extension of $G_U^{(R)}$.}
\label{tab:Pin+_U-duality_Summary}
\end{table}

\subsection{Decompactification Limits}

In this section we study the decompactification limit from the $D=7$ U-duality group to the $D=8$ U-duality group.
The bosonic U-duality groups are
\begin{equation}
G_U^\mathrm{7D} = \SL(5,\mathbb{Z}) \quad \text{and} \quad G_U^\mathrm{8D} = \SL(3,\mathbb{Z}) \times \SL(2,\mathbb{Z}),
\end{equation}
respectively. In the bosonic case, the decompactification limit is simple, as
$\SL(3,\mathbb{Z}) \times \SL(2,\mathbb{Z})$ embeds block diagonally in $\SL(5,\mathbb{Z})$.
However, additional subtleties arise when considering the $\Pin^+$-lift. In particular, there is an $\SL(2,\mathbb{Z})$ simple factor of $G_U^\mathrm{8D}$ which we can geometrically interpret as the group of large diffeomorphisms of a $T^2$ arising from type II string theory on $T^2$. Type IIB string theory is not well defined on non-orientable manifolds, so the lifted duality group should never have orientation-reversing elements belonging \emph{solely} to $\GL(2,\mathbb{Z})$. We will show that this is indeed the case by studying the decompactification limit from 7D.

To be precise, start with the block diagonal embedding\footnote{We argue for this particular embedding by noting that $\SL(3,\mathbb{Z})$ and $\SL(2,\mathbb{Z})$ are each subgroups of $\SL(5,\mathbb{Z})$ that need to commute with each other.} 
\begin{equation}
\iota:\; \SL(3,\mathbb{Z})\times \SL(2,\mathbb{Z}) \hookrightarrow \SL(5,\mathbb{Z}),\qquad
\iota(A,B)=\begin{pmatrix}A&0\\[4pt]0&B\end{pmatrix}, \label{eq:Decomposition}
\end{equation}
realizing $\mathbb{Z}^5=\mathbb{Z}^3\oplus\mathbb{Z}^2$.  At the purely bosonic level one can move a diagonal sign among coordinates by conjugation in the full group. For instance, with
\begin{equation}
R=\operatorname{diag}(-1,1,1,1,1),\qquad
P=\begin{pmatrix}
0&0&0&0&1\\[2pt]
1&0&0&0&0\\[2pt]
0&1&0&0&0\\[2pt]
0&0&1&0&0\\[2pt]
0&0&0&1&0
\end{pmatrix}\in \SL(5,\mathbb{Z}),
\end{equation}
one has $P^4 R P^{-4}=\operatorname{diag}(1,1,1,1,-1)$.  Crucially, $P$ is not block-diagonal with respect to the chosen decomposition, so this conjugation does not stay inside the block-diagonal subgroup.

Now include the reflection / Pin$^+$-lift as
\begin{equation}
1 \rightarrow \mathbb{Z}_2 \rightarrow \widetilde{G_U}^+ \rightarrow
SL(5,\mathbb{Z})\rtimes\mathbb{Z}_{2}^R \rightarrow 1,
\end{equation}
and let $H:=\bigl(\SL(3,\mathbb{Z})\times \SL(2,\mathbb{Z})\bigr)\rtimes\mathbb{Z}_{2}$ be the restricted semi-direct subgroup.  Pulling back the extension gives $\iota^*\widetilde{G_U}^+$, the unique $\Pin^+$-lift of the subgroup as it sits inside the full lattice automorphism group.

The obstruction is now immediate and unavoidable: any conjugation that would relocate the $-1$ into the $\SL(2,\mathbb{Z})$-block requires a conjugator $P$ outside the block-diagonal normalizer, and lifting that conjugation to $\widetilde{G_U}^+$ will insert a central sign. In short, the bosonic group admits coordinate moves by non-block conjugation. Thus, once the reflection and Pin$^+$ central sign are included, the block-diagonal subgroup cannot contain a pure reflection on the $\SL(2,\mathbb{Z})$ factor, as the semi-direct structure prevents it.\footnote{One may wonder if the issue may arise if we had instead started with a reflection element in the $\SL(2,\mathbb{Z}$) sub-block of $\SL(5,\mathbb{Z})$ in \eqref{eq:Decomposition}. However, we see that this element will not only extend the $\SL(2,\mathbb{Z})$ block, but a combination of the the $\SL(3,\mathbb{Z})$ and $\SL(2,\mathbb{Z})$ blocks due to the semi-direct product.}

\section{Lyndon–Hochschild–Serre Spectral Sequence} \label{app:LHS}

In our analysis we examine various extensions of groups, by the addition of a reflection as well as the introduction of Spin- and Pin$^+$-lifts. The group cohomology of these extensions, which we use for the analysis of Abelianizations and bordism groups, can be accessed via the Lyndon-Hochschild-Serre (LHS) spectral sequence. In particular, for an extension
\begin{equation}
\label{GNext}
    1 \rightarrow N \rightarrow \widetilde{G} \rightarrow G \rightarrow 1 \,,
\end{equation}
the second page of the homological LHS spectral sequence is given by
\begin{equation}
\label{the_LHS}
    E^2_{p,q} = H_p \big( B G; H_q (BN; A)\big) \Longrightarrow H_{p+q} (B \widetilde{G}; A) \,,
\end{equation}
where in general $A$ is a $G$-module, and in all cases $H_q(BN; A)$ has a $\widetilde G$-module structure arising from the $G$-action on $N$ induced by~\eqref{GNext}. In our case it will suffice to consider $A = \mathbb{Z}$ or $A = \mathbb{Z}_2$ with trivial $G$-action. Thus, the group homology of $\widetilde{G}$ can be related to the group homology of $G$ and $N$.

\subsection{Inclusion of Reflections}
\label{subapp:homref}

Let us first include the reflection in the bosonic U-duality group, describing a semi-direct product with $\mathbb{Z}_2^R$. In this subsection, we make the following computations.
\begin{prop}
\label{refl_coh}
Suppose the dimension $D\le 7$.
\begin{enumerate}
    \item\label{h1_refl} Let $A$ be either $\mathbb Z$ or $\mathbb Z_2$. Then $H_1(B(G_U\rtimes\mathbb{Z}_2^R); A)\cong\mathbb Z_2$.
    \item\label{h2_inj} The map $H_2(BG_U;\mathbb Z)\to H_2(B(G_U\rtimes\mathbb Z_2^R);\mathbb Z)$ is an isomorphism.
\end{enumerate}
\end{prop}
\begin{proof}
We have a short exact sequence
\begin{equation}
    1 \rightarrow G_U \rightarrow G_U \rtimes \mathbb{Z}_2^R \rightarrow \mathbb{Z}_2^R \rightarrow 1 \,,
\end{equation}
so we can apply the LHS spectral sequence~\eqref{the_LHS}:
\begin{equation}
\label{U_dual_LHS}
    E^2_{p,q} = H_p \big( B \mathbb{Z}_2^R; H_q (BG_U; \mathbb{Z}) \big) \Longrightarrow H_{p+q} \big(B(G_U \rtimes \mathbb{Z}_2^R); \mathbb{Z} \big) \,.
\end{equation}
Since $D \leq 7$, $G_U$ is perfect, i.e.,
\begin{equation}
    \text{Ab} \big[G_U\big] = H_1 (BG_U; \mathbb{Z}) = 0 \,.
\end{equation}
We also have
\begin{subequations}\label{Z2_hom}
\begin{equation}
H_k (B\mathbb{Z}_2; \mathbb{Z}) = \begin{cases}
    \mathbb{Z} & k =0 \\
    \mathbb{Z}_2 & k>0,\, k \text{ odd} \\
    0 & \text{otherwise}\end{cases}
\end{equation}
and
\begin{equation}
H_k (B \mathbb{Z}_2; \mathbb{Z}_2) = \mathbb{Z}_2 \,, \quad k \geq 0 \,.
\end{equation}
\end{subequations}
Thus the $E^2$-page of the LHS spectral sequence~\eqref{U_dual_LHS} for $p,q$ small is
\begin{equation}\label{refl_E2}
    \begin{array}{c | c c c c}
    2 & H_2(BG_U;\mathbb Z) & H_2(BG_U;\mathbb Z_2) & 0 & H_2(BG_U;\mathbb Z_2)\\
    1 & 0 & 0 & 0 & 0\\
    0 & \mathbb{Z} & \mathbb{Z}_2 & 0 & \mathbb Z_2\\ \hline
    q / p & 0 & 1 & 2 & 3
    \end{array}
\end{equation}
Thus in degrees $1$ and below, degree considerations mean there are no nonzero differentials nor extension problems and this spectral sequence collapses at the $E^2$-page. Thus
\begin{equation}
    \text{Ab} \big[ G_U \rtimes \mathbb{Z}_2^R \big] = H_1 \big( B (G_U \rtimes \mathbb{Z}_2^R); \mathbb{Z} \big) \cong \mathbb{Z}_2.
\end{equation}
This takes care of part~\eqref{h1_refl} for $A = \mathbb Z$; the proof for $\mathbb{Z}_2$ coefficients is essentially the same. For part~\eqref{h2_inj}, looking at the $E^2$-page~\eqref{refl_E2}, since total degree $2$ otherwise vanishes, it suffices to prove that $E^2_{0,2}\cong H_2(BG_U;\mathbb Z)$ survives to the $E^\infty$-page. The only differential to or from it that does not vanish for degree reasons is $d_3\colon E^2_{3,0}\to E^2_{0,2}$. However, in the LHS spectral sequence for a semidirect product, all differentials which cross the line $q = 0$ must vanish. This is because the quotient map $q\colon G_U\rtimes\mathbb Z_2^R\to \mathbb Z_2^R$ has a section given by a choice of reflection, so the pushforward map $q_*\colon H_*(B(G_U\rtimes\mathbb Z_2^R);\mathbb Z)\to H_*(B\mathbb Z_2^R)$ also has a section, hence must be surjective. This pushforward map is the edge homomorphism in the spectral sequence, meaning that it is realized as the quotient by all elements with $q>0$; for this to be surjective, differentials cannot kill any classes on the line $q = 0$. Therefore the $d_3$ of interest vanishes and the inclusion $G_U\to G_U\rtimes\mathbb Z_2^R$ is indeed an isomorphism on $H_2$.
\end{proof}

\subsection{Spin- and \texorpdfstring{Pin$^+$}{Pin+}-Lifts}

In this section we turn to the analogous statements incorporating Spin- and Pin$^{+}$-Lifts.

\subsubsection{Spin-Lift}

Next, we discuss the non-trivial $\mathbb{Z}_2$ central extension associated to the Spin-lift of $G_U$ and described by the short exact sequence
\begin{equation}
\label{spin_lift_xtn}
    1 \rightarrow \mathbb{Z}_2 \rightarrow \widetilde{G_U} \rightarrow G_U \rightarrow 1 \,.
\end{equation}
\begin{prop}
\label{spin_lift_H1}
For $D\le 7$, $H_1(B\widetilde{G_U}; \mathbb Z) = \mathrm{Ab}\big[\widetilde{G_U}\big] = 0$.
\end{prop}
We will once again use the LHS spectral sequence, but this time we will need to know the value of a differential. To do so, we will want to work in a more general setting. Let
\begin{equation}
\label{the_Serre_gen}
    B\mathbb Z_2\to Y\to X
\end{equation}
be a principal $B\mathbb Z_2$-bundle, which is classified by a map $f\colon X\to B^2\mathbb Z_2$. Since $B^2\mathbb Z_2$ is a $K(\mathbb Z_2, 2)$, the homotopy class of $f$ is equivalent data to a class $w\in H^2(X;\mathbb Z_2)$. This generalizes the case of a central extension 
\begin{equation}
\label{to_be_Serre}
1 \to \mathbb Z_2\to \widetilde G\to G\to 1;
\end{equation}
set $X = BG$ and $Y = B\widetilde G$. Then the classifying space functor turns the extension~\eqref{to_be_Serre} into the fibration~\eqref{the_Serre_gen}, and identifies the LHS spectral sequence for this extension with the Serre spectral sequence for the fibration. The class $w\in H^2(BG;\mathbb Z_2)$ equals the cohomology class classifying the central extension~\eqref{to_be_Serre}.
\begin{lem}
\label{LHS_d2}
Given $X$, $Y$, $f$, and $w$ as above, such that $X$ is connected, then in the homological Serre spectral sequence for~\eqref{the_Serre_gen}, which has signature
\begin{equation}\label{Serre_sig}
    E^2_{p,q} = H_p(X; H_q(B\mathbb Z_2; \mathbb Z)) \Longrightarrow H_{p+q}(Y;\mathbb Z),
\end{equation}
the differential\footnote{Typically the differential is denoted by $d^2$ in the homological version of the spectral sequence, but we stick to $d_2$ here in order to avoid confusion with the square of the differential.} $d_2\colon E^2_{2,0}\to E^2_{0,1}$ is the map
\begin{equation}
\label{diffform}
    \begin{aligned}
        H_2(X;\mathbb Z) &\longrightarrow H_0(X;\mathbb Z_2)\cong\mathbb Z_2\\
        x &\longmapsto w\frown (x\bmod 2).
    \end{aligned}
\end{equation}
\end{lem}
Here, ``$\frown$'' denotes the cap product, which evaluates a cohomology class on a homology class. Sometimes this can be calculated using the universal coefficient theorem, which gives us a short exact sequence
\begin{equation}
\label{UCT_pair}
    0\to \mathrm{Ext}(H_{n-1}(X; \mathbb Z), \mathbb Z_2) \to
    H^n(X;\mathbb Z_2) \overset{h}{\to}
    \mathrm{Hom}(H_n(X;\mathbb Z), \mathbb Z_2)\to 0
\end{equation}
such that the map $h$ sends a cohomology class $y$ to the homomorphism $x\mapsto y\frown (x\bmod 2)$.

\begin{proof}
The classifying map $f$ induces a map of $E^2$-pages of spectral sequences from~\eqref{Serre_sig} to the Serre spectral sequence for the universal principal $B\mathbb Z_2$-bundle
\begin{equation}
\label{univ_fib}
    1\to B\mathbb Z_2 \to E(B\mathbb Z_2) \to B^2\mathbb Z_2\to 1.
\end{equation}
This map of spectral sequences commutes with differentials, and under the identification of the $E^2$-page with homology, this map is the pushforward map on homology. Let $'E^r_{p,q}$ denote the Serre spectral sequence associated to~\eqref{univ_fib}.
Then we obtain a commutative diagram
\begin{equation}
\begin{tikzcd}
	{E^2_{2,0} = H_2(X;\mathbb Z)} && {'E_{2,0}^2 = H_2(B^2\mathbb Z_2; \mathbb Z)} & {} \\
	{E^2_{0,1} = H_0(X;\mathbb Z_2)} && {'E_{0,1}^2 = H_0(B^2\mathbb Z_2; \mathbb Z_2)} \\
	{} & {\mathbb Z_2.}
	\arrow["{f_*}", from=1-1, to=1-3]
	\arrow["{d_2}", from=1-1, to=2-1]
	\arrow["{d_2}", from=1-3, to=2-3]
	\arrow["{f_*}", from=2-1, to=2-3]
	\arrow["\cong", from=2-1, to=3-2]
	\arrow["\cong"', from=2-3, to=3-2]
\end{tikzcd}
\end{equation}
The claimed formula for $d_2$ in~\eqref{diffform} also commutes with these maps, so it suffices to prove the lemma in the case of the universal principal $B\mathbb Z_2$-bundle~\eqref{univ_fib}, for which $f = \mathrm{id}$ and $w$ is the tautological class.

The Hurewicz theorem implies $H_2(B^2\mathbb Z_2;\mathbb Z)\cong\mathbb Z_2$; let $y$ be the nonzero element. Since $H_1(B^2\mathbb Z_2;\mathbb Z) = 0$ because $B^2\mathbb Z_2$ is simply connected, the universal coefficient long exact sequence~\eqref{UCT_pair} collapses to an isomorphism $h\colon H^2(B^2\mathbb Z_2;\mathbb Z_2)\to\mathrm{Hom}(H_2(B^2 \mathbb Z_2;\mathbb Z), \mathbb Z_2)$. This map sends $w$ to the homomorphism ``cap product with $w$,'' so $w\frown (y\bmod 2) =1\in\mathbb Z_2$, since $y\bmod 2$ is the only class on which that cap product can be nonzero. Since $B^2\mathbb Z_2$ is connected, $H_0(B^2\mathbb Z_2;\mathbb Z_2)\cong\mathbb Z_2$. Therefore $d_2(y) = w\frown(y\bmod 2)$ if and only if this $d_2$ is nonzero. But this $d_2$ must be nonzero: since $E(B\mathbb Z_2)$ is contractible, $E^\infty_{p,q}$ must vanish if $p+q>0$, and since $E^2_{0,1}$ can only be killed by this $d_2$, this $d_2$ is nonzero.
\end{proof}
\begin{proof}[Proof of \cref{spin_lift_H1}]
We study the homological LHS spectral sequence with $\mathbb Z$ coefficients associated to~\eqref{spin_lift_xtn}, which has $E^2$-page
\begin{equation}
\label{general_spin_lift_LHS}
\begin{array}{c | c c c c}
2 & 0 & 0 & 0 & 0 \\
1 & H_0(BG_U; \mathbb{Z}_2) & H_1(BG_U; \mathbb{Z}_2) & H_2( BG_U; \mathbb{Z}_2) & H_3(BG_U; \mathbb{Z}_2) \\
0 & H_0(BG_U; \mathbb{Z}) & H_1(BG_U; \mathbb{Z}) & H_2( BG_U; \mathbb{Z}) & H_3(BG_U; \mathbb{Z}) \\ \hline
q/p & 0 & 1 & 2 & 3
\end{array}
\end{equation}
where we already plugged in the homology groups of $B\mathbb{Z}_2$ from~\eqref{Z2_hom}. Since we have restricted to $D \leq 7$, so that $G_U$ is perfect, \eqref{general_spin_lift_LHS} simplifies to
\begin{equation}
\begin{array}{c | c c c c}
2 & 0 & 0 & 0 & 0 \\
1 & \mathbb{Z}_2 & 0 & H_2( BG_U; \mathbb{Z}_2) & H_3(BG_U; \mathbb{Z}_2) \\
0 & \mathbb{Z} & 0 & H_2( BG_U; \mathbb{Z}) & H_3(BG_U; \mathbb{Z}) \\ \hline
q/p & 0 & 1 & 2 & 3
\end{array}
\end{equation}
From this we can see that $H_0 (B\widetilde{G}_U; \mathbb{Z}) = \mathbb{Z}$, and $H_1 (B\widetilde{G}_U; \mathbb{Z})$ is either $0$ or $\mathbb Z_2$ if $d_2\colon E^2_{2,0}\to E^2_{0,1}$ is nonzero, resp.\ $0$. By \cref{LHS_d2}, for any class $x\in H_2(BG_U;\mathbb Z)$, $d_2(x) = w\frown (x\bmod 2)$.

Now use the universal coefficient theorem again. Since $G_U$ is perfect, $H_1(BG_U;\mathbb Z) = 0$, so just as in the proof of \cref{LHS_d2}, the universal coefficient short exact sequence~\eqref{UCT_pair} collapses to an isomorphism $h\colon H^2(BG_U;\mathbb Z_2)\to\mathrm{Hom}(H_2(B^2 \mathbb Z_2;\mathbb Z), \mathbb Z_2)$, where $h(y)$ is the function $x\mapsto y\frown(x\bmod 2)$. The class $w\in H^2(BG_U;\mathbb Z_2)$ representing the extension $\widetilde{G_U}\to G_U$ is nonzero, because in all examples of interest, this extension is nonsplit. Therefore the function $d_2(x) = x\mapsto w\frown (x\bmod 2)$ is also nonzero, which implies $H_1(B\widetilde{G_U};\mathbb Z) = 0$.
\end{proof}
As a consistency check, this matches what we observed in bordism in \S\ref{ss:spin_lift_bordism}. The argument in the proof of \cref{spin_lift_H1} also fixes the second homology group to be
\begin{equation}
    H_2 (B\widetilde{G_U}; \mathbb{Z}) = \text{ker}(d_2) \subset H_2 (BG_U; \mathbb{Z}) \,,
\end{equation}
which in the case of $\SL(n,\mathbb{Z})$ vanishes.

\subsubsection{\texorpdfstring{Pin$^+$-Lift}{Pin+ Lift}}

Finally, we analyze the homology for the full Pin$^+$-lift of U-duality groups in $D\leq7$. This lift is described by a short exact sequence
\begin{equation}
\label{pinp_LES}
1 \rightarrow \mathbb{Z}_2 \rightarrow \widetilde{G_U}^+ \rightarrow G_U \rtimes \mathbb{Z}_2^R \rightarrow 1 \,.
\end{equation}
\begin{prop}
\label{pinp_LHS}
For $D\le 7$, $H_1(B\widetilde{G_U}^+;\mathbb Z) \cong \mathrm{Ab}\big[\widetilde{G_U}^+\big] \cong\mathbb Z_2$.
\end{prop}
\begin{proof}
Set up the homological LHS spectral sequence associated to~\eqref{pinp_LES} with $\mathbb Z$ coefficients. The $E^2$-page is
\begin{equation}
   \begin{array}{c | c c c}
   2 & 0 & 0 & 0 \\
   1 & \mathbb{Z}_2 & H_1 \big( B(G_U \rtimes \mathbb{Z}_2^R); \mathbb{Z}_2\big) & H_2 \big( B(G_U \rtimes \mathbb{Z}_2^R); \mathbb{Z}_2\big)  \\
   0 & \mathbb{Z} & \mathbb{Z}_2 & H_2\big( B(G_U \rtimes \mathbb{Z}_2^R); \mathbb{Z}\big).  \\ \hline
   q/p & 0 & 1 & 2
   \end{array}
\end{equation}
The $\mathbb Z_2$ summand in $E^2_{1,0}$ survives to the $E^\infty$-page for degree reasons, so we are done if we can show that $d_2\colon E^2_{2,0}\to E^2_{0,1}$ is nonzero, to remove the other $\mathbb Z_2$ summand in total degree $1$. For this, consider the map of short exact sequences induced by the homomorphism $j\colon\widetilde{G_U}\hookrightarrow\widetilde{G_U}^+$ including the Spin-lift of $G_U$ into the $\Pin^+$-lift:
\begin{equation}
\begin{tikzcd}
	1 & {\mathbb Z_2} & {\widetilde{G_U}} & {G_U} & 1 \\
	1 & {\mathbb Z_2} & {\widetilde{G_U}^+} & {G_U\rtimes\mathbb Z_2^R} & 1
	\arrow[from=1-1, to=1-2]
	\arrow[from=1-2, to=1-3]
	\arrow[equals, from=1-2, to=2-2]
	\arrow[from=1-3, to=1-4]
	\arrow[from=1-3, to=2-3]
	\arrow[from=1-4, to=1-5]
	\arrow["j", from=1-4, to=2-4]
	\arrow[from=2-1, to=2-2]
	\arrow[from=2-2, to=2-3]
	\arrow[from=2-3, to=2-4]
	\arrow[from=2-4, to=2-5]
\end{tikzcd}
\end{equation}
This induces a map of LHS spectral sequences which on the $E^2$-page is the pushforward map $j_*$ on homology, and which commutes with all differentials. For the time being, let $E_{p,q}^r$ denote the LHS for the Spin-lift and $^+E_{p,q}^r$ denote the LHS for the $\Pin^+$-lift. Thus we have a commutative diagram
\begin{equation}
\label{LHS_compare}
\begin{tikzcd}
	{E^2_{2,0} = H_2(BG_U;\mathbb Z)} && {{}^+E_{2,0}^2 = H_2(B(G_U\rtimes\mathbb Z_2); \mathbb Z)} & {} \\
	{E^2_{0,1} = H_0(BG_U;\mathbb Z_2)} && {{}^+E_{0,1}^2 = H_0(B(G_U\rtimes\mathbb Z_2); \mathbb Z_2)} \\
	{} & {\mathbb Z_2.}
	\arrow["{j_*}", from=1-1, to=1-3]
	\arrow["{d_2}", from=1-1, to=2-1]
	\arrow["{^+d_2}", from=1-3, to=2-3]
	\arrow["{j_*}", from=2-1, to=2-3]
	\arrow["\cong", from=2-1, to=3-2]
	\arrow["\cong"', from=2-3, to=3-2]
\end{tikzcd}
\end{equation}
In the proof of \cref{spin_lift_H1} we saw that the leftmost $d_2$ is surjective; since $j_*$ is an isomorphism in degree $0$, then ${}^+d_2\circ j_*$ (i.e.\ traveling along the upper right of~\eqref{LHS_compare}) is also surjective. Therefore ${}^+d_2$ must also be surjective, so as we mentioned above, we have finished showing $H_1(B\widetilde{G_U}^+;\mathbb Z)\cong\mathbb Z_2$.
\end{proof}
Like for the Spin-lift, \cref{pinp_LHS} matches the results we obtained in bordism in \S\ref{ss:pinp_lift_bordism}.

\section{Calculation of \texorpdfstring{$\Omega_1^{\Spin\text{-}\widetilde{G_U}} (\mathrm{pt})$}{first spin-twisted bordism} and \texorpdfstring{$\Omega_1^{\Spin\text{-}\widetilde{G_U}^+}(\mathrm{pt})$}{first pin-twisted bordism}}
\label{app:Adams}
In this Appendix we calculate $\Omega_1^{\Spin\text{-}\widetilde{G_U}}(\text{pt})$ and $\Omega_1^{\Spin\text{-}\widetilde{G_U}^+}(\text{pt})$ for the U-duality groups $G_U$ in dimensions $3\le D\le 9$. We summarize the answers in the following theorems.
\begin{thm}
\label{spin_lift_bordism}
For $3\le D\le 9$, there is an isomorphism
\begin{equation}
    \Omega_1^{\Spin\text{-}\widetilde{G_U}} (\mathrm{pt})\cong \mathrm{Ab}\big[\widetilde{G_U}\big] \cong
        \begin{cases}
            0, & 3\le D\le 7\\
            \mathbb Z_{24}, &D = 8, 9.
        \end{cases}
\end{equation}
\end{thm}
\begin{thm}
\label{pinp_lift_bordism}
For $3\le D\le 9$, there is an isomorphism
\begin{equation}
    \Omega_1^{\Spin\text{-}\widetilde{G_U}^+}(\mathrm{pt}) \cong \mathrm{Ab}\big[\widetilde{G_U}^+\big] \cong
        \begin{cases}
            \mathbb Z_2, & 3\le D\le 7\\
            \mathbb Z_2\oplus\mathbb Z_2, &D = 8,9.
        \end{cases}
\end{equation}
In all cases, one $\mathbb Z_2$ summand is generated by the bordism class of a circle whose duality bundle has monodromy given by the non-trivial element of $\mathbb Z_2^R$; in $D = 8,9$ the other $\mathbb Z_2$ summand is represented by a circle with monodromy given by a Spin-lift of $S\in\SL(2, \mathbb Z)$ to $\Mp(2, \mathbb Z)$.
\end{thm}
We also discuss the bosonic version of the spin result; this is not a new result, but offers a nice parallel to \cref{spin_lift_bordism}.
\begin{prop}
\label{bosonic_bord}
For $3\le D\le 9$, there is an isomorphism
\begin{equation}
    \Omega_1^{\Spin}(BG_U) \cong \mathbb Z_2\oplus \mathrm{Ab}\big[G_U\big] \cong
        \begin{cases}
            \mathbb Z_2, & 3\le D\le 7\\
            \mathbb Z_2\oplus \mathbb Z_{12}, &D = 8,9.
        \end{cases}
\end{equation}
In all cases, one $\mathbb Z_2$ summand is generated by the circle with periodic spin structure and trivial duality bundle; in $D = 8,9$, the other $\mathbb Z_{12}$ summand is generated by a circle with either spin structure and a duality bundle whose monodromy is $T\in\SL(2, \mathbb Z)$.
\end{prop}
We will prove these theorems by two different methods: an Atiyah-Hirzebruch spectral sequence calculation in \S\ref{ss:AHSS}, and an Adams spectral sequence calculation in \S\ref{ss:Adams}. We focus on the case $D\le 7$ to simplify the arguments. Adams and Atiyah-Hirzebruch spectral sequence calculations for the $D = 9$ cases of \cref{spin_lift_bordism,pinp_lift_bordism,bosonic_bord} appear in~\cite[\S A]{Dierigl:2020lai} and~\cite[\S\S 12--14]{Debray:2023yrs}. This leaves $D = 8$, which we briefly discuss at the end of \S\ref{ss:AHSS}.

\subsection{Atiyah-Hirzebruch Spectral Sequence Calculations}
\label{ss:AHSS}
The Atiyah-Hirzebruch spectral sequence for the reduced Spin bordism of a space $X$ has signature
\begin{equation}
\label{untwisted_AHSS}
E^2_{p,q} = \widetilde H_p(X; \Omega_q^\Spin(\pt)) \Longrightarrow \widetilde \Omega_{p+q}^\Spin(X).
\end{equation}

\begin{proof}[Proof of \cref{bosonic_bord} using the Atiyah-Hirzebruch spectral sequence]
Apply the Atiyah-Hirzebruch spectral sequence~\eqref{untwisted_AHSS} with $X = BG_U$, and recall $\Omega_0^\Spin(\pt)\cong\mathbb Z$ and $\Omega_1^\Spin(\pt)\cong\mathbb Z_2$. On the $E^2$-page, in total degree $1$, we have
\begin{subequations}
  \begin{align}
    E^2_{1,0} &\cong \widetilde H_1(BG_U;\Omega_0^\Spin(\pt)) \cong \widetilde H_1(BG_U;\mathbb Z) = \mathrm{Ab}\big[G_U\big]\\
    E^2_{0,1} &\cong \widetilde H_0(BG_U; \Omega_1^\Spin(\pt)) \cong \widetilde H_0(BG_U;\mathbb Z_2)\cong 0.
  \end{align}
\end{subequations}
For degree reasons, all differentials into or out of $E^2_{1,0}$ vanish, so it survives intact to the $E^\infty$-page. There is no extension problem, so $\widetilde\Omega_1^\Spin(BG_U)\cong \mathrm{Ab}\big[G_U\big]$. For unreduced Spin bordism, $\Omega_1^\Spin(BG_U)\cong\widetilde\Omega_1^\Spin(BG_U)\oplus\Omega_*^\Spin(\pt)$, so we direct-sum on $\Omega_1^\Spin(\pt)\cong\mathbb Z_2$.
\end{proof}
For the Spin- and Pin-lifts, we must use a twisted variant.
\begin{defn}[{Wang~\cite[Definition 8.2]{Wan08}}]
\label{twspin}
Let $X$ be a space and $w\in H^2(X;\mathbb Z_2)$. An $(X, w)$-twisted spin structure on an oriented vector bundle $E\to M$ is the data of a map $f\colon M\to X$ and a trivialization of $w_2(E) + f^*(w)$.\footnote{There is a more general notion of twisted spin structure allowing modifications of both $w_1$ and $w_2$: see Hebestreit-Joachim~\cite{HJ20} as well as~\cite[Definition 1.24]{Debray:2023tdd}. The James spectral sequence, and the formulas for its low-degree $d_2$s, both exist in this generality, as is proven in Kasprowski-Powell~\cite[Proposition 3.9]{KP25} following Teichner~\cite{Tei97}; see also~\cite[\S 5]{OP25}. The Adams spectral sequence we use in \S\ref{ss:Adams} for twisted Spin bordism also generalizes to this setting; see~\cite{Debray:2023tdd}.}
\end{defn}
We will let $\Omega_k^\Spin(X, w)$ denote the group of bordism classes of $k$-dimensional manifolds with $(X, w)$-twisted spin structures.
\begin{lem}[\!\!{\cite[\S 3.1]{Debray:2023tdd}}]
\label{twist_is_twist}
Let $w\in H^2(BG;\mathbb Z_2)$ be the cohomology class of the extension $1\to\mathbb Z_2\to\widetilde{G}\to G\to 1$. Then the notions of a $\Spin\text{-}\widetilde G$ structure and a $(BG, w)$-twisted spin structure on a vector bundle are canonically equivalent.
\end{lem}
Thus, $\Omega_*^{\Spin\text{-}\widetilde{G_U}}\cong \Omega_*^\Spin(BG_U, w)$ and $\Omega_*^{\Spin\text{-}\widetilde{G_U}^+}\cong \Omega_*^\Spin(B(G_U\rtimes\mathbb Z_2^R), w)$. Here we are implicitly using that the class $w\in H^2(BG_U;\mathbb Z_2)$ is invariant under the $\mathbb Z_2^R$-action, therefore passes through the LHS spectral sequence to define a class in $H^2(B(G_U\rtimes\mathbb Z_2^R);\mathbb Z_2)$, which we also call $w$. Alternatively, the Spin double cover $\widetilde{G_U}\to G_U$ extends to the $\Pin^+$ double cover $\widetilde{G_U}^+\to G_U\rtimes\mathbb Z_2^R$, so the cohomology class $w$ classifying it is the restriction of a class $w\in H^2(B(G_U\rtimes\mathbb Z_2^R);\mathbb Z_2)$. When we write ``$w$,'' it will always be clear from context which of these two classes we mean.\footnote{In addition to the $\Pin^+$ lift of $G_U\rtimes\mathbb Z_2^R$, there is also a $\Pin^-$ lift $\widetilde{G_U}{}^-$, classified by $w + r^2\in H^2(B(G_U\rtimes\mathbb Z_2);\mathbb Z_2)$. Here $r$ is the pullback of the unique nonzero class in $H^1(B\mathbb Z_2^R;\mathbb Z_2)$ by the quotient map $G_U\rtimes\mathbb Z_2^R\to\mathbb Z_2^R$.}
\begin{thm}[{Teichner~\cite[Proposition 1]{Tei93}}]
\label{jamesSS}
With $(X, w)$ as in \cref{twspin}, there is a spectral sequence with signature
\begin{subequations}
\label{jamessig}
    \begin{equation}\label{jamesE2}
        E^2_{p,q} = H_p(X; \Omega_q^\Spin(\pt)) \Longrightarrow \Omega_{p+q}^\Spin(X, w).
    \end{equation}
    such that, at least for $2\le p\le 4$, the differential $d_2\colon E^2_{p,0}\to E^2_{p-2,1}$, as a map $H_p(X; \mathbb Z)\to H_{p-2}(X;\mathbb Z_2)$, is reduction modulo $2$ followed by the dual of the map
    \begin{equation}
    \label{Sq2w}
      \begin{aligned}
        \Sq^2_w \colon H^{p-2}(X;\mathbb Z_2) &\to H^p(X;\mathbb Z_2)\\
                 x &\mapsto \Sq^2x + wx.
      \end{aligned}
    \end{equation}
\end{subequations}
\end{thm}
``Dual'' in \cref{jamesSS} means: the universal coefficient theorem shows that the cap product pairing canonically identifies $H_k(X;\mathbb Z_2)$ and $H^k(X;\mathbb Z_2)$ as dual $\mathbb Z_2$-vector spaces, and we take the dual of the linear map $\Sq^2_w$.

Teichner calls the spectral sequence~\eqref{jamessig} the James spectral sequence. See~\cite{Tei93, Boh03, Olb08, Ped17} for some example computations with this spectral sequence similar to those in this paper.

Teichner shows that if there is an orientable vector bundle $V\to X$ with $w_2(V) = w$, then the James spectral sequence is isomorphic to the Atiyah-Hirzebruch spectral sequence for the Spin bordism of the Thom spectrum $X^{V - \mathrm{rank}(V)}$ (that this is $\Omega_*^\Spin(X, w_2(V))$-twisted Spin bordism is a folk theorem, with one proof given in \cite[Corollary 10.19]{Debray:2023yrs}). It is not always possible to realize every degree-$2$ cohomology class $w$ as $w_2$ of a vector bundle~\cite[\S 2]{GKT89}, and indeed this issue can occur when $w$ is the extension class for the Spin-lift of a real U-duality group $G_U(\mathbb R)$~\cite[Theorem 4.2]{DY24}, so the extra generality of the James spectral sequence is necessary.\footnote{More general twisted Atiyah-Hirzebruch spectral sequences have been constructed and discussed in~\cite{MS06, BN19, GS17, GS19a, GS19, BM21, GS22}, but the differentials we need in the case of twisted Spin bordism have not been computed, so we use the James spectral sequence.}
\begin{proof}[Proof of \cref{spin_lift_bordism} for $D\le 7$ using the James spectral sequence]
We begin by drawing the $E^2$-page of this spectral sequence~\eqref{jamesE2} for $(BG_b, w)$-twisted Spin bordism, which by \cref{twist_is_twist} also computed $\Spin\text{-}\widetilde{G_U}$ bordism. Since $D\le 7$, $G_U$ is perfect and so $H_1(BG_U;\mathbb Z)$ and $H_1(BG_U;\mathbb Z_2)$ both vanish.
\begin{equation}
   \begin{array}{c | c c c}
   1 & \mathbb{Z}_2 & 0 & H_2 \big( BG_U; \mathbb{Z}_2\big)  \\
   0 & \mathbb{Z} & 0 & H_2\big( BG_U; \mathbb{Z}\big).  \\ \hline
   q/p & 0 & 1 & 2
   \end{array}
\end{equation}
The only nonzero group in total degree $1$ is $E^2_{0,1}\cong\mathbb Z_2$, and this is the target of $d_2\colon E^2_{2,0}\to E^2_{0,1}$. So if we can show that this $d_2$ is nonzero, we are done. In \cref{jamesSS}, we learned that $d_2 = (\Sq^2_w)^\vee\circ r$, where $r$ is reduction mod $2$; we will show $(\Sq^2_w)^\vee$ and $r$ are both surjective, which implies their composition is nonzero.

For $r$, surjectivity follows from the Bockstein long exact sequence
\begin{equation}
    \dotsb\to H_2(BG_U; \mathbb Z) \overset r\to H_2(BG_U;\mathbb Z_2) \to H_1(BG_U;\mathbb Z) = 0\to\dotsb
\end{equation}
The codomain of $(\Sq^2_w)^\vee$ is isomorphic to $\mathbb Z_2$, so this map is surjective if and only if it is nonzero, which is true if and only if its dual $\Sq_w^2\colon H^0(BG_U;\mathbb Z_2)\to H^2(BG_U;\mathbb Z_2)$ is nonzero. And indeed:
\begin{equation}
\label{sq2w1}
    \Sq_w^2(1) = \Sq^2(1) + w\cdot 1 = 0 +w = w,
\end{equation}
as $\Sq^2$ vanishes in degree $0$ of any space. Since $\widetilde{G_U}\to G_U$ is a nonsplit extension, $w\ne 0$ and so we are done.
\end{proof}
\begin{proof}[Proof of \cref{pinp_lift_bordism} for $D\le 7$ using the James spectral sequence]
The proof for the $\Pin^+$-lift is similar. By \cref{twist_is_twist}, we want to calculate $(B(G_U\rtimes\mathbb Z_2^R), w)$-twisted Spin bordism. To draw the $E^2$-page, we recall from \cref{refl_coh} that, since $D\le 7$, $H_1(BG_U;\mathbb Z)$ and $H_1(BG_U;\mathbb Z_2)$ are both isomorphic to $\mathbb Z_2$.
\begin{equation}
   \begin{array}{c | c c c}
   1 & \mathbb{Z}_2 & \mathbb Z_2& H_2 \big( B(G_U\rtimes \mathbb Z_2^R); \mathbb{Z}_2\big)  \\
   0 & \mathbb{Z} & \mathbb Z_2 & H_2\big( B(G_U\rtimes \mathbb Z_2^R); \mathbb{Z}\big).  \\ \hline
   q/p & 0 & 1 & 2
   \end{array}
\end{equation}
Once again we are done if we can show that $d_2\colon E^2_{2,0}\to E^2_{0,1}$ is nonzero; this kills $E^2_{0,1}$, and the $\mathbb Z_2$ in $E^2_{1,0}$ survives to the $E^\infty$-page for degree reasons, so this would imply that in total degree $1$ on the $E^\infty$-page, there is a single $\mathbb Z_2$ summand and no other nonzero classes, which would finish the proof.

We will compute this differential by comparing it with the differential for $(BG_U, w)$-twisted Spin bordism. Specifically, because the inclusion map $j\colon G_U\to G_U\rtimes\mathbb Z_2^R$ pulls $w\in H^2(B(G_U\rtimes\mathbb Z_2^R);\mathbb Z_2)$ back to $w\in H^2(BG_U;\mathbb Z_2)$, there is a map of James spectral sequences commuting with differentials, analogous to the map of LHS spectral sequences that we used in the proof of \cref{pinp_LHS}.

Specifically, if we let $E^r_{p,q}$ denote the James spectral sequence for $(BG_U, w)$ and ${}^+E^r_{p,q}$ denote the James spectral sequence for $(B(G_U\rtimes\mathbb Z_2^R), w)$, then we have a commutative diagram
\begin{equation}
\begin{tikzcd}
	{E^2_{2,0} = H_2(BG_U;\mathbb Z)} & {} & {{}^+E^2_{2,0} = H_2(BG_U;\mathbb Z)} \\
	{E^2_{0,1} = H_0(BG_U;\mathbb Z_2)} && {{}^+E^2_{0,1} = H_0(B(G_U\rtimes\mathbb Z_2^R);\mathbb Z_2)} \\
	& {\mathbb Z_2}
	\arrow["{j_*}", from=1-1, to=1-3]
	\arrow["{d_2}", two heads, from=1-1, to=2-1]
	\arrow["{d_2}", from=1-3, to=2-3]
	\arrow["{j_*}", from=2-1, to=2-3]
	\arrow["\cong"{description}, from=2-1, to=3-2]
	\arrow["\cong"{description}, from=2-3, to=3-2]
\end{tikzcd}
\end{equation}
Since both classifying spaces are connected, the pushforward map $j_*$ on $H_0$ is an isomorphism. We proved in \cref{refl_coh} that $j_*$ is also an isomorphism on $H_2$. In the proof of \cref{spin_lift_bordism} by the James spectral sequence earlier in this subsection, we showed that the left-hand vertical map, the $d_2$ for $(BG_U, w)$, is surjective. Commutativity thus implies the right-hand vertical map, which is the $d_2$ of interest, is also surjective. As noted above, this finishes the proof.
\end{proof}
Thus only dimensions $D = 8,9$ are left. These fall to the James spectral sequence in a similar manner; we highlight a few differences and leave the details to the interested reader. For the Spin lifts:
\begin{itemize}
    \item In these examples, $d_2\colon E^2_{2,0}\to E^2_{0,1}$ is $0$, so in total degree $1$, the $E^\infty$ page consists of a $\mathbb Z_2$ in $E^2_{0,1}$ and a $\mathbb Z_{12} = \mathrm{Ab}\big[G_U\big]$ in $E^\infty_{1,0}$.
    \item There is a hidden extension joining $E^\infty_{1,0}$ and $E^\infty_{0,1}$. This can be seen by embedding $i\colon \mathbb Z_4\hookrightarrow\SL(2, \mathbb Z)$ (or into $\SL(2, \mathbb Z)\times\SL(3, \mathbb Z)$) and using the map of spectral sequences. The presence of a non-trivial hidden extension in this degree in $(B\mathbb Z_4, i^*w)$-twisted Spin bordism follows from~\cite[\S 13.4]{Debray:2023yrs}.
\end{itemize}
For the $\Pin^+$-lifts, the story is similar: $E^2_{1,0}\cong\mathbb Z_2$, and again the differential vanishes. This time, the extension splits, as can be seen by embedding $D_{16}\hookrightarrow \SL(2, \mathbb Z)\rtimes\mathbb Z_2^R$ or $(\SL(2, \mathbb Z)\times\SL(3, \mathbb Z))\rtimes\mathbb Z_2^R$ and comparing with~\cite[Theorem 14.18]{Debray:2023yrs}.

\subsection{Adams Spectral Sequence Calculations}
\label{ss:Adams}\label{app:AdamsSpin}
Lastly, we use the Adams spectral sequence to compute the one-dimensional $\Spin\text{-}\widetilde{G_U}$ and $\Spin\text{-}\widetilde{G_U}^+$ bordism groups. This tool is more abstract than the Atiyah-Hirzebruch spectral sequence, but in the last several years has become more prominent in the theoretical physics literature as a powerful yet tractable way to compute twisted Spin bordism groups. See~\cite{beaudry2018guidecomputingstablehomotopy,Debray:2023yrs} for introductions to this technique aimed at a mathematical physics audience as well as several example computations.

The references cited apply the Adams spectral sequence to the $(X, w)$-twisted Spin bordism under the assumption that there is an oriented vector bundle $V\to X$ with $w_2(V) = w$. As we noted above in \S\ref{ss:AHSS}, not only is this not true in general, there are specific counterexamples for the Spin-lifts of the real versions of U-duality groups. Thus, following~\cite{Debray:2023tdd}, we use Baker-Lazarev's relative Adams spectral sequence~\cite{BL01} in this section. See~\cite{Debray:2023tdd, Kur25} for example computations with this version of the Adams spectral sequence.

Let $\cA(1)$ be the subalgebra of the Steenrod algebra generated by $\Sq^1$ and $\Sq^2$. Since $\Sq^1$ and $\Sq^2$ act naturally on mod $2$ cohomology groups, the mod $2$ cohomology of any space is naturally an $\cA(1)$-module.
\begin{lem}[{\!\!\cite[Lemma 2.27(3)]{Debray:2023tdd}}]
\label{BLA}
Let $X$ be a space, $w\in H^2(X;\mathbb Z_2)$, and $\Sq_w^2$ be the operator defined in~\eqref{Sq2w}. Then the actions of $\Sq^1$ and $\Sq_w^2$ on $H^*(X;\mathbb Z_2)$ satisfy the Adem relations for $\Sq^1$ and $\Sq^2$, and therefore define another $\cA(1)$-module structure on $H^*(X;\mathbb Z_2)$ which we call $H_w^*(X;\mathbb Z_2)$.
\end{lem}
That is: if we modify the action of $\Sq^2$ on $H^*(X;\mathbb Z_2)$ by having it act by $\Sq_w^2$ instead, the result is still a well-defined $\cA(1)$-module, and we call that $\cA(1)$-module $H_w^*(X;\mathbb Z_2)$.
\begin{thm}[{\!\!\cite[\S 2.2]{Debray:2023tdd}}]
Let $X$ be a space and $w\in H^2(X;\mathbb Z_2)$. Then there is a graded $\cA(1)$-module $M$ and a spectral sequence
\begin{equation}
    E_2^{s,t} = \mathrm{Ext}_{\cA(1)}^{s,t}(H_w^*(X;\mathbb Z_2)\otimes M, \mathbb Z_2) \Longrightarrow \Omega_{t-s}^\Spin(X, w) _2^\wedge.
\end{equation}
There is an $\cA(1)$-module map $M\to\mathbb Z_2$ which is an isomorphism in degrees $7$ and below.
\end{thm}
Here $\mathrm{Ext}$ is the derived functor of $\mathrm{Hom}$ (see~\cite[\S 11.2]{Debray:2023yrs}) and $(\text{--})_2^\wedge$ denotes ``$2$-completion.'' For the spaces we consider in this paper, whose Spin bordism groups are finitely generated Abelian groups, $2$-completion can heuristically be thought of as keeping the free and $2$-torsion summands and throwing out the odd-torsion summands. Thus in this subsection we will do something else to account for odd-primary torsion.

For any space $X$ and class $w\in H^2(X;\mathbb Z_2)$, the map $\phi\colon \Omega_*^\Spin(X, w)\to\Omega_*^\SO(X)$ is an isomorphism after localizing at any odd prime. This means that $\phi$ is an isomorphism on odd-torsion subgroups. When $w = 0$, this is a standard fact (see, e.g., \cite[\S 10.5]{Debray:2023yrs}); for general $w$ it is less well-known but still a folklore theorem. See~\cite[Proof of Lemma 3.23]{DYY25} for a proof.

There is a map $\psi\colon \Omega_*^\SO(X)\to H_*(X;\mathbb Z)$ obtained by sending an oriented manifold $M$ with map $f\colon M\to X$ to $f_*([M])$, and Thom's computation of $\Omega_*^\SO$ in low degrees~\cite[Théorème IV.13]{ThomThesis} implies $\psi$ is an isomorphism in degrees $3$ and below. Therefore to compute the odd-primary torsion in $\Omega_1^{\Spin\text{-}\widetilde{G_U}}$ and $\Omega_1^{\Spin\text{-}\widetilde{G_U}^+}$, it suffices to know $H_1(BG_U;\mathbb Z)$ and $H_1(B(G_U\rtimes \mathbb Z_2^R), \mathbb Z)$, i.e.\ the Abelianizations of these groups, which we discussed above, e.g.\ in \cref{refl_coh}. Thus in what follows we will ignore odd torsion.

Before we get into the proofs of \cref{spin_lift_bordism,pinp_lift_bordism,bosonic_bord}, we have one more simplification to discuss. To run the Adams spectral sequence, we need to compute the Ext groups of $H_w^*(BG_U;\mathbb Z_2)$ and $H_w^*(B(G_U\rtimes\mathbb Z_2^R);\mathbb Z_2)$, but we know very little about these cohomology groups in degrees $2$ and above. Fortunately, graded Ext is ``local'' in the sense that low-degree information in cohomology completely determines Ext in low topological degrees. We have two points of view on this phenomenon.
\begin{enumerate}
    \item Suppose you want to compute $\mathrm{Ext}_{\cA(1)}^{s,t}(N, \mathbb Z_2)$ with a minimal $\cA(1)$-module resolution $P_\bullet\to N$ (see~\cite[\S 4.4]{beaudry2018guidecomputingstablehomotopy}). In practice, if one only knows the structure of $N$ in degrees $d$ and below, a minimal resolution can be constructed explicitly in degrees $t-s<d$, simply by trying to work out the minimal resolution and stopping once higher-degree information on $N$ is necessary. This technique works because of a theoretical guarantee that, with one exception we discuss in a moment, if $N$ is concentrated in nonnegative degrees, the lowest-degree class in $P_s$, the $s^{\mathrm{th}}$ step of the minimal resolution, is approximately $3s$~\cite{Liu63, Ada64}. The exception is $h_0$-towers, which are not hard to recognize in a minimal resolution and so can be accounted for (see~\cite[Example 4.4.2]{beaudry2018guidecomputingstablehomotopy}).
    \item Alternatively, we can use a long exact sequence in Ext associated to a short exact sequence of $\mathcal{A}(1)$ modules, as in~\cite[\S 4.6]{beaudry2018guidecomputingstablehomotopy}. Let $N$ be an $\mathcal{A}(1)$ module and $N_{>\ell}$ the submodule of $N$ consisting of elements in degrees $\ell +1$ and above. Then there is a short exact sequence
\begin{equation}
    0 \rightarrow N_{>\ell} \rightarrow N \rightarrow N/N_{>\ell} \rightarrow 0 \,,
\end{equation}
which induces a long exact sequence in Ext. Since $N_{>\ell}$ is only concentrated in degree $\ell+1$ and above, $\mathrm{Ext}_{\cA(1)}(N_{>\ell})$ is concentrated in degree $t-s \geq \ell +1$. Exactness then implies that
\begin{equation}
    \mathrm{Ext}_{\cA(1)}(N) \rightarrow \mathrm{Ext}_{\cA(1)}(N/N_{>\ell}) \,
\end{equation}
is an isomorphism for $t-s \leq \ell$. In our approach $\ell = 1$ due to our limited knowledge of the group homology of the U-duality groups.
\end{enumerate}
Thus in particular, we may completely ignore the $\cA(1)$-module $M$ appearing in \cref{BLA}, as it cannot affect the behavior of the spectral sequence in degrees $6$ and below.\footnote{In fact, because the Anderson-Brown-Peterson splitting of Spin bordism~\cite{ABP67} generalizes to $(X, w)$-twisted Spin bordism~\cite[\S 5.2]{HJ20}, this extends to degree $7$.}
\begin{proof}[Proof of \cref{bosonic_bord} for $D\le 7$ using the Adams spectral sequence]
We want to compute untwisted Spin bordism of $BG_U$, i.e.\ twisted Spin bordism for the class $0\in H^2(BG_U;\mathbb Z_2)$. Therefore the input to Ext is simply $H^*(BG_U;\mathbb Z_2)$ as an $\cA(1)$-module, as $\Sq^2_w$ for $w = 0$ equals $\Sq^2$.

Because $\widetilde{G_U}\to G_U$ is a non-trivial extension, its class in $H^2$ is nonzero, so there is at least one factor of $\mathbb{Z}_2$ in $H^2(BG_U; \mathbb{Z}_2)$. Therefore we can completely determine $H^*(BG_U;\mathbb Z_2)$ as an $\cA(1)$-module in degrees $2$ and below: we have the class $1\in H^0$ and some number of classes in $H^2$ (and nothing in $H^1$ because $G_U$ is perfect), and there is no action of the Steenrod algebra that connects $H^0 (BG_U; \mathbb{Z}_2)$ with any element in $H^2 (BG_U; \mathbb{Z}_2)$ (simply because $\Sq^1(1) = \Sq^2 (1) = 0)$. This means that when we apply Ext to $H^0 (BG_U; \mathbb{Z})$, it simply produces the usual Spin bordism contributions of $\Omega^{\Spin}_k (\pt)$ and we can focus on the positive-degree part of $H^*(BG_U;\mathbb Z_2)$, which we call $X$. The first step in a minimal resolution is
\begin{equation}
X \longleftarrow \Sigma^2 \mathcal{A}(1)^{\oplus n} \oplus \Sigma^3 R \,,
\end{equation}
where $n$ is given in terms of
\begin{equation}
H^2 (BG_U; \mathbb{Z}_2) = \mathbb{Z}_2^{\oplus n} \,,
\end{equation}
and $R$ denotes an unknown remainder. For the next step we only obtain new copies of $\mathcal{A}(1)$ in suspension $\Sigma^3$ and higher, in the next only in suspension $\Sigma^4$ and so on:
\begin{equation}
X \longleftarrow \Sigma^2 \mathcal{A}(1)^{\oplus n} \oplus \Sigma^3 R \longleftarrow \Sigma^3 \tilde{R} \longleftarrow \Sigma^4 \tilde{R}' \longleftarrow \dots \,,
\end{equation}
where the particular form of $\tilde{R}$ and $\tilde{R}'$ depends on the details we do not know. However, this is enough to draw the Adams chart for the first two columns, which are completely empty for $X$. From this we deduce the well-known fact that for perfect groups $G$,
\begin{equation}
\Omega^{\Spin}_1 (BG_U) = \Omega^{\Spin}_1 (\pt) \oplus \text{Ab}[G_U] =  \Omega^{\Spin}_1 (\pt) = \mathbb{Z}_2 \,.
\label{eq:notwist}
\end{equation}
Thus $\Omega_*^\Spin(BG_U)$ only receives contributions from the Spin bordism of a point.
\end{proof}
\begin{proof}[Proof of \cref{spin_lift_bordism} for $D\le 7$ using the Adams spectral sequence]
This time around, we must use $\Sq^2_w$ with $w\ne 0$. This acts non-trivially from degree $0$ to degree $2$: as we saw in~\eqref{sq2w1}, $\Sq_w^2(1) = w$, which modifies the minimal resolution:
\begin{equation}
H_w^*(BG_U;\mathbb Z_2) \longleftarrow \mathcal{A}(1) \oplus \mathcal{A}(1)^{\oplus n-1} \oplus \Sigma^3 Q \longleftarrow \Sigma^1 \mathcal{A}(1) \oplus \Sigma^3 \widetilde{Q} \longleftarrow \Sigma^2 \mathcal{A}(1) \oplus \Sigma^4 \widetilde{Q}' \longleftarrow \dots \,,
\end{equation}
where we denote the unknown contributions by $Q$, $\widetilde{Q}$, and $\widetilde{Q}'$, respectively.
\begin{figure}
\centering
\includegraphics[width = 0.7 \textwidth]{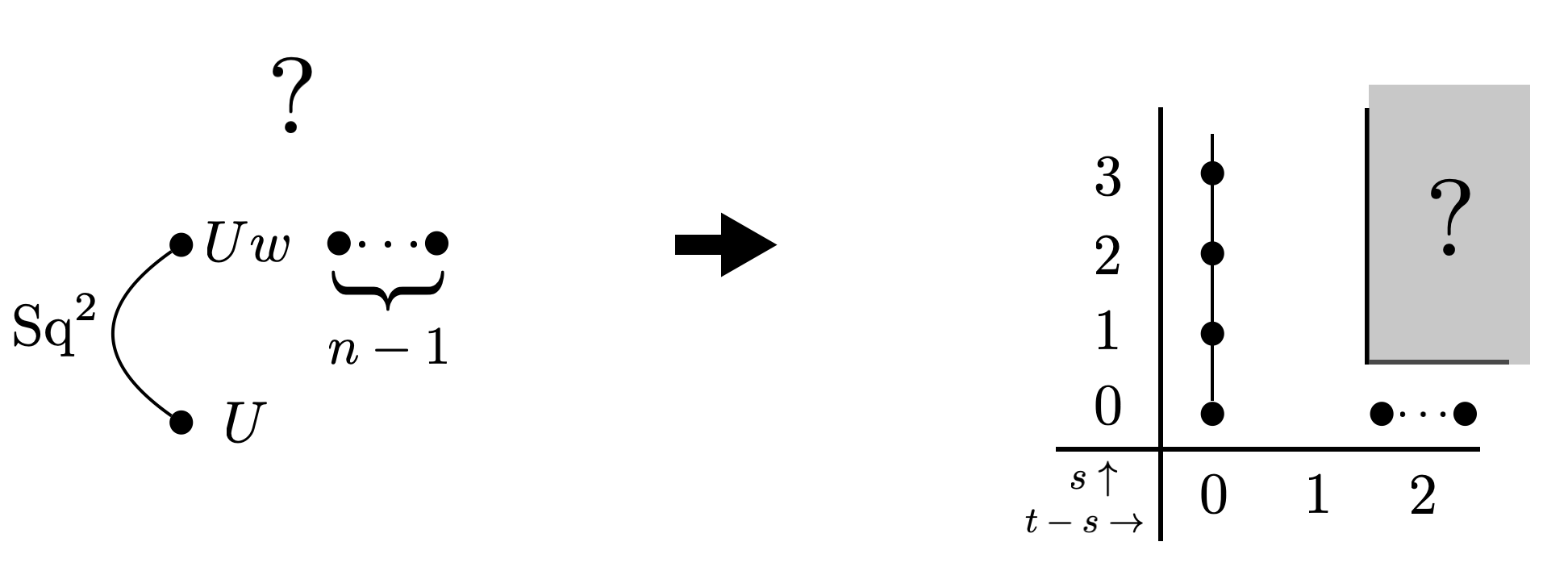}
\caption{Steenrod structure of twisted Spin structure for the Spin-lift of perfect U-duality groups (left) and associated Adams chart (right).}
\label{fig:AdamsSpin}
\end{figure}
This leads to an Adams chart with the first column empty (see Figure \ref{fig:AdamsSpin}), from which we read off
\begin{equation}
\Omega^{\Spin\text{-}\widetilde{G_U}}_1 (\pt) = 0 \,,
\end{equation}
irrespective of the details.
\end{proof}
\begin{proof}[Proof of \cref{pinp_lift_bordism} for $D\le 7$ using the Adams spectral sequence]
Let $r\in H^1(B(G_U\rtimes\mathbb Z_2^R);\mathbb Z_2)$ be the pullback of the unique nonzero class in $H^1(B\mathbb Z_2^R;\mathbb Z_2)$ by the quotient map $q\colon G_U\rtimes\mathbb Z_2^R\to\mathbb Z_2^R$, and let $w$ be the $H^2$ class of the extension $\widetilde{G_U}^+\to G_U\rtimes\mathbb Z_2^R$. Then we have:
\begin{equation}
\Sq^1(r) = r^2 \,, \quad \Sq^2 (r) = 0 \,, \quad \Sq^2 (w) = w^2 \,, \quad \Sq^1 (w) = r w + w_3 (V) \,,
\end{equation}
which follow from the axioms and the Wu formula. Further note that $r^2$ is non-trivial: because the quotient $q$ has a section given by a choice of reflection, $H^*(B\mathbb Z_2^R; \mathbb Z_2)$ is a direct summand of $H^*(B(G_U\rtimes\mathbb{Z}_2^R); \mathbb Z_2)$. Since $r^2\ne 0$ in $H^*(B\mathbb Z_2^R; \mathbb Z_2)$, the same is true pulled back to $B(G_U\rtimes\mathbb{Z}_2^R)$. Thus, we get a picture for the $\cA(1)$-module structure on $H_w^*(B(G_U\rtimes\mathbb Z_2^R);\mathbb Z_2)$ in degrees two and lower:
\begin{equation}
\Sq^1 (1) = 0 \,, \enspace \Sq^2_w (1) = w \,, \enspace \Sq^1 (r) = r^2 \,, \enspace \Sq^2_w (r) = w r \,, 
\end{equation}
If $R_2$ denotes the $\cA(1)$-module which is the kernel of the unique non-trivial $\cA(1)$-module map $\Sigma^{-1}\cA(1)\to\Sigma^{-1}\mathbb Z_2$ (see~\cite[Figure 28, left]{beaudry2018guidecomputingstablehomotopy} for a picture), then in degrees $2$ and below, $H_w^*(B(G_U\rtimes\mathbb Z_2^R);\mathbb Z_2)$ is isomorphic to the quotient of $R_2$ by its degree-$\ge 3$ elements.\footnote{For the nine-dimensional U-duality group, $H_w^*(B(G_U\rtimes\mathbb Z_2^R);\mathbb Z_2)$ is computed in degrees $12$ and below in~\cite[Proposition 14.21]{Debray:2023yrs}, and the entire $R_2$ summand is visible.
There is also another element in degree 1, which was associated to rotations; this part is absent here, since the U-duality groups for $D\le 7$ are perfect and the Abelianization of $(G_U \rtimes \mathbb{Z}_2^R)$ is given by $\mathbb{Z}_2^R$ only.}
We can therefore write down the first few steps in a minimal resolution for $H_w^*(B(G_U\rtimes\mathbb Z_2^R);\mathbb Z_2)$:
\begin{equation}
H_w^*(B(G_U\rtimes\mathbb Z_2^R);\mathbb Z_2) \longleftarrow \mathcal{A}(1) \oplus \Sigma \mathcal{A}(1) \oplus \Sigma^2 P \longleftarrow \Sigma^3 \widetilde{P} \longleftarrow \Sigma^4 \widetilde{P}' \longleftarrow \dots \,,
\end{equation}
with unknowns $P$, $\widetilde{P}$, $\widetilde{P}'$.
\begin{figure}
\centering
\includegraphics[width = 0.6 \textwidth]{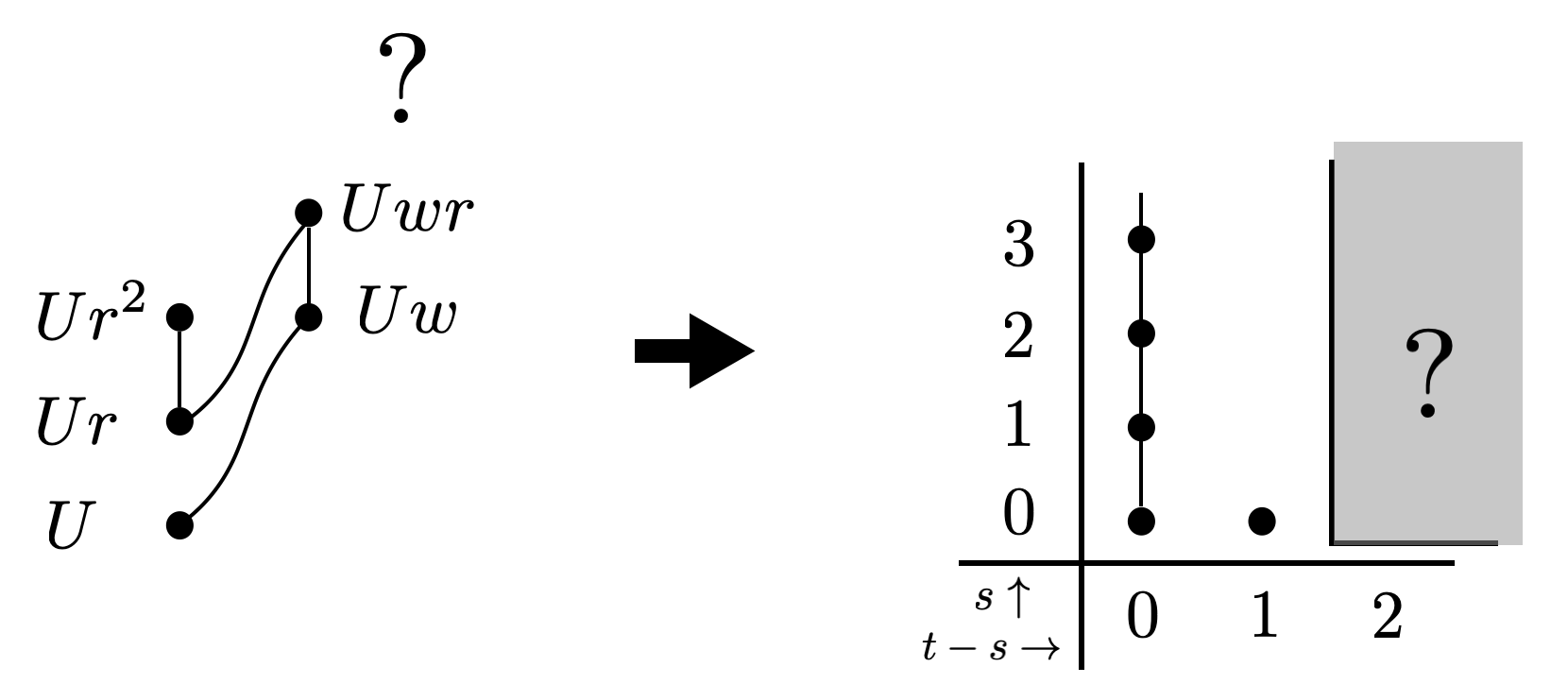}
\caption{Steenrod structure of Pin$^+$-lift of perfect U-duality groups (left) and associated Adams chart (right).}
\label{fig:AdamsPin}
\end{figure}
This is enough to determine the the first two columns of the Adams chart as summarized in Figure \ref{fig:AdamsPin}. From this we read off
\begin{equation}
\Omega^{\Spin\text{-}\widetilde{G_U}^{+}}_0 ( \pt ) = \mathbb{Z} \,, \quad \Omega^{\Spin\text{-}\widetilde{G_U}^{+}}_1 ( \text{pt} ) = \mathbb{Z}_2 \,.
\end{equation}
The $\mathbb{Z}_2$ is associated to the reflection element $r$, hence in twisted Spin bordism is represented by a circle with a duality bundle whose monodromy is a reflection.
\end{proof}

\newpage

\bibliographystyle{utphys}
\bibliography{Pintopia}

\end{document}